\documentclass[11pt]{article}

\usepackage{amsfonts}
\usepackage{amsmath}
\usepackage{amssymb}
\usepackage{lscape}
\usepackage{float}
\usepackage{graphicx}
\usepackage{theorem}

\newcommand{\qed}{$\hfill\Box$}
\newcommand{\R}{\mathbb{R}}
\newcommand{\E}{\mathbb{E}}
\newcommand{\PP}{\mathbb{P}}
\newcommand{\f}{\mathcal F}
\newcommand{\Perp}{\perp\!\!\!\perp}

\newcommand{\al}{\alpha}
\newcommand{\be}{\beta}

\newcommand{\si}{\sigma}

\newcommand{\De}{\Delta}

\newcommand{\Om}{\Omega}

\newcommand{\vp}{\varphi}

\newcommand{\proba}{(\Omega ,\mathcal{F},(\f_t)_{t\in [0,1]},\PP)}
\newcommand{\proban}{(\Omega^{(0)} ,\mathcal{F}^{(0)},(\f_t^{(0)})_{t\in [0,1]},\PP^{(0)})}
\newcommand{\probae}{(\Omega^{(1)} ,\mathcal{F}^{(1)},(\f_t^{(1)})_{t\in [0,1]},\PP^{(1)})}

\newcommand{\toop}{\stackrel{\PP}{\longrightarrow}}
\newcommand{\schw}{\stackrel{d}{\longrightarrow}}
\newcommand{\eqschw}{\stackrel{d}{=}}
\newcommand{\stab}{\stackrel{d_{st}}{\longrightarrow}}

\newcommand{\bee}{\begin{equation}}
\newcommand{\eee}{\end{equation}}
\newcommand{\bea}{\begin{eqnarray}}
\newcommand{\eea}{\end{eqnarray}}
\newcommand{\bean}{\begin{eqnarray*}}
\newcommand{\eean}{\end{eqnarray*}}

\renewcommand{\theequation}{\arabic{section}.\arabic{equation}}

\newtheorem{prop}{Proposition}[section]
\newtheorem{cor}[prop]{Corollary}

\newtheorem{lem}[prop]{Lemma}

\newtheorem{theo}[prop]{Theorem}
\newtheorem{rem}[prop]{Remark}

\begin{document}

\title{On covariation estimation for multivariate continuous It\^o
semimartingales with noise in non-synchronous observation schemes}
\author{Kim Christensen \thanks{CREATES, University of Aarhus, Building 1326,
8000 Aarhus, Denmark, e-mail: kchristensen@creates.au.dk.} \and Mark Podolskij
\thanks{Department of Mathematics, Heidelberg University, INF 294, 69120
Heidelberg, Germany, e-mail: m.podolskij@uni-heidelberg.de.} \and Mathias Vetter
\thanks{Ruhr-Universit\"at Bochum, Fakult\"at f\"ur Mathematik, 44780 Bochum,
Germany, e-mail: mathias.vetter@rub.de}}

\date{March, 2013} \maketitle

\begin{abstract}
This paper presents a Hayashi-Yoshida type estimator for the covariation matrix
of continuous It\^o semimartingales observed with noise. The coordinates of the
multivariate process are assumed to be observed at highly frequent
non-synchronous points. The estimator of the covariation matrix is designed via
a certain combination of the local averages and the Hayashi-Yoshida estimator.
Our method does not require any synchronization of the observation scheme (as
e.g. previous tick method or refreshing time method) and it is robust to some
dependence structure of the noise process. We show  the associated central
limit theorem for the proposed estimator and provide a feasible asymptotic
result. Our proofs are based on a blocking technique and a stable convergence
theorem for semimartingales. Finally, we show simulation results for the proposed
estimator to illustrate its finite sample properties.


\end{abstract}

\thispagestyle{empty}
\newpage

\section{Introduction} \label{Real-Intro}
\setcounter{page}{1} In the past years there has been a considerable development
of statistical methods for stochastic processes observed at  high frequency.
This was mainly motivated by financial applications, where the data,
such as stock prices or currencies, are observed very frequently.
It is well known that under the no-arbitrage assumption price processes must follow
a semimartingale (see e.g. \cite{DS}). However, at ultra high frequencies the financial
data is contaminated by {\it microstructure noise} such as rounding errors, bid-ask bounds and misprints.
This fact prevents us from using classical power variation based methods
(see e.g. \cite{BGJPS} or \cite{J2} among many others) to infer the characteristics of a semimartingale.

A standard model for a continuous It\^o semimartingale observed with errors is given by
\begin{equation} \label{ypr}
Y_t=X_t + \varepsilon_t, \qquad t\geq 0,
\end{equation}
where $(X_t)_{t\geq 0}$ is a $d$-dimensional process ({\it true price}) of the form
\begin{equation} \label{xpr}
X_t= X_0 + \int_0^t a_s ds + \int_0^t \sigma_s dW_s, \qquad t\geq 0,
\end{equation}
with $(a_s)_{s\geq 0}$ being an $\R^d$-valued c\`agl\`ad process, $(\sigma_s)_{s\geq 0}$ being an
$\R^{d\times d^{\prime}}$-valued c\`agl\`ad volatility and $W$ representing a $d^{\prime}$-dimensional Brownian motion, and
the $d$-dimensional error process $\varepsilon$ ({\it microstructure noise}) is iid with
\bean
\E[\varepsilon_t]=0, \qquad \E[\varepsilon_t \varepsilon_t^{\star}]= \Psi \in \R^{d\times d},
\eean
independent of $X$. Throughout this work an asterisk denotes the transpose of a matrix.

The aim of this paper is to estimate the covariation matrix of $X$ over some interval, say $[0,1]$, i.e.
\bean
[X]= \int_0^1 \Sigma_s ds\in \R^{d \times d}, \qquad \Sigma_s= \sigma_s \sigma_s^{\star},
\eean
based on non-synchronous noisy observations ($Y=(Y^1, \ldots, Y^d)$)
\begin{equation*}
Y_{t_i^k}^k, \qquad k=1, \ldots,d, \qquad i=0,\ldots, n_k,
\end{equation*}
where $0=t_0^k < \cdots<t_{n_k}^k=1$ are partitions of the interval $[0,1]$ with $\max_{1\leq i\leq n_k}|t_i^k - t_{i-1}^k|\rightarrow 0$
as $n_k\rightarrow \infty$ for all $1\leq k\leq d$. The univariate counterpart of this problem has been studied intensively in the literature. Let us mention the {\it two-scale approach} of \cite{ZMA} (see \cite{Z} for its more efficient multi-scale version),
the {\it realised kernel method} proposed in \cite{BHLS} and the {\it pre-averaging concept} originally introduced in \cite{PV2}
(and further studied in \cite{JLMPV}, \cite{JPV}, \cite{PV1} in various settings) among others. These methods can be extended to the
multivariate case in a rather straightforward manner if the observations are synchronous.

When the underlying data is non-synchronous, things are less obvious, as we are faced with two challenges at the same time: We have to de-noise the data as before, but we also need to apply a certain synchronization technique to create a new set of observations from which appropriate estimators for $[X]$ can be computed. For the multivariate realised kernel method, \cite{BHLS2} proposed to cope with non-synchronous data by applying the {\it refreshing time method} first, which synchronizes the observations via a previous tick method. In a second step, a noise robust estimator is constructed from this new data set. Similar in spirit is the extension of the multi-scale estimator due to \cite{B}, where synchronous observations are obtained using the \textit{pseudo-aggregation algorithm} of \cite{PA} first. The resulting covariance estimator then becomes a multi-scale version of the Hayashi-Yoshida estimator from \cite{HY}, which originally has been introduced to deal with non-synchronicity in semimartingale models without noise.

Both approaches have their drawbacks, however: (a) Using the previous tick approach (which generates pseudo data points) may lead
to inconsistent estimators for certain observation schemes; this phenomenon has been noticed in \cite{HY} in the setting
of a pure diffusion; (b) After any of the synchronization techniques there remain at most $\min_{1\leq k\leq d} (n_k)$ data points,
which amounts in throwing away a lot of data. In the no-noise case, this is usually no problem, as for the Hayashi-Yoshida estimator exactly those observations are dropped that bear no additional information on the covariance, but for noisy data they still can be used to wipe out the noise.

To avoid these afore-mentioned drawbacks, we propose to combine a synchronization technique and a concept for de-noising as well, but in reverse order: We apply the pre-averaging approach, which is designed to locally diminish the influence of the noise, first, and use the Hayashi-Yoshida method afterwards. Our estimator, denoted by $HY^n$, has the following important properties: \\ \\
(i) In general, we use all observations $Y_{t_i^k}^k$; \\ \\
(ii) The estimator has the optimal convergence rate $n^{-1/4}$; \\ \\
(iii) The estimation method is robust to certain dependence structures of the noise process. This property is important for
practical applications as the economic theory typically does not provide any insight on modeling the noise. \\ \\
The main idea of the construction of $HY^n$ comes from \cite{CKP}, where we indicated its consistency, but did not provide
the complete asymptotic theory. In this paper we now
prove a stable central limit theorem for $HY^n - [X]$. From a technical point of view, the conditions we use on the observation scheme $t_i^k$ are rather mild, but on the other hand there is no empirical evidence that such assumptions are reasonable in financial practice. However, a thorough analysis involving e.g.\ random observations times is beyond the scope of our paper. Furthermore,
we explain how to estimate the (random) asymptotic covariance matrix that appears in the central limit theorem to obtain
a {\it feasible} result (which may be used in practice to construct confidence regions). We would like to emphasize again that
the construction of our estimator is not completely obvious (as there are several ways of combining the Hayashi-Yoshida method
and the pre-averaging approach, which may result in different properties) and that the proof of the main result, which is based
on a certain blocking technique, martingale inequalities and a stable central limit theorem for semimartingales, is  more advanced
than in the univariate setting.

This paper is organized as follows: in Section \ref{sec1} we
introduce the set up and explain the construction of $HY^n$. The main results of the paper including the consistency
of $HY^n$ and the associated stable central limit theorems are presented in Section \ref{sec2}. Section \ref{varest} deals with estimation techniques for the conditional variance, while in Section \ref{sec3}
we show some numerical results
to illustrate the finite sample properties of our estimator. Section \ref{sec4} is devoted to proofs, and some tedious parts are relegated to an Appendix in Section \ref{append}.

\section{The set up} \label{sec1}
\setcounter{equation}{0}
\renewcommand{\theequation}{\thesection.\arabic{equation}}

We start by introducing an appropriate filtered probability space on which our noisy process $Y$ is defined. Let \newline
$\proban$ be an arbitrary space on which the true price process $X$ lives, such that all involved process $a$, $\sigma$
and $W$ are adapted. Now we consider a second filtered probability space $\probae$, where $\Omega^{(1)}$ is the set of functions
from $[0,1]$ to $\R^d$ and $\mathcal F^{(1)}$ is the product $\sigma$-field of the Borel $\sigma$-algebras $\mathcal A_t$ on $\R^d$, indexed by $t \in [0,1]$. We define on it the noise process
$\varepsilon =(\varepsilon_t)_{t\in [0,1]}$ as follows: let $Q$ be a probability law on $\R^d$ (the marginal law of $\varepsilon$)
and set $\PP^{(1)}$ as $\PP^{(1)}=\otimes_{t\in [0,1]} P_t$
with $P_t=Q$ for all $t\in [0,1]$. Now, $(\varepsilon_t)_{t\in [0,1]}$ is defined as
the canonical process on $\probae$ with $(\mathcal F^{(1)}_t)_{t\in [0,1]}$ being
the canonical filtration. The process $Y$ in (\ref{ypr}) lives on the product space $\proba$ given by:
\bean
\Om~=~\Om^{(0)}\times\Om^{(1)},\qquad
\f~=~\f^{(0)}\times\f^{(1)},\qquad
\f_t~=~\f^{(0)}_t\otimes~\f^{(1)}_t,\qquad \PP=\PP^{(0)} \otimes \PP^{(1)}.
\eean
We remark that the probability space on which the process $\varepsilon$ lives is rather minimal; a precise definition of it is required for the stable
convergence results, however. The process $Y$ is defined in continuous time just for convenience, although the mapping $(\omega,t)\rightarrow
Y_t(\omega )$ is not $\mathcal F \otimes \mathcal B([0,1])$-measurable.\newline \newline
Now we introduce the assumptions on the sampling scheme. \\ \\
{\it Assumption (T):} The observation times $t_i^k$, $i=0,\ldots, n_k,$ $k=1, \ldots, d$ satisfy the following conditions:
\begin{itemize}
\item[(T1)] ({\it Time transformation})
$t_i^k$'s are transformations of an equidistant grid, i.e. there exist strictly monotonic (deterministic) functions
$f_k:[0,1]\rightarrow [0,1]$ in $C^1([0,1])$ with non-zero right and left derivative in 0 and 1, respectively, and with $f_k(0)=0$, $f_k(1)=1$ such that
\begin{equation} \label{ttrans}
t_i^k=f^{-1}_k (i/n_k), \qquad i=0,\ldots, n_k, \quad k=1, \ldots, d.
\end{equation}
\item[(T2)] ({\it Boundedness of $f'_k$})  There exists a natural number $M>0$ such that
\bean
M^{-1}< \inf_{x\in [0,1]} f'_{k}(x) \leq \sup_{x\in [0,1]} f'_{k}(x) <M, \qquad k=1, \ldots, d.
\eean
\item[(T3)] ({\it Comparable number of observations}) Set $n=\sum_{k=1}^d n_k$. It holds that
\begin{equation} \label{nk}
\frac{n_k}{n}\rightarrow m_k\in (0,1], \qquad k=1, \ldots, d.
\end{equation}
\item[(T4)] ({\it Joint grid points})
The grids $(t_i^{k})$, $(t_j^{l})$ ($1\leq k,l\leq d$) have $n_{kl}$ common points which are denoted by
$(t_{p}^{kl})_{1\leq p\leq n_{kl}}$. They have the representation $t_{p}^{kl}=f_{kl}^{-1} (p/n_{kl})$ and $
n_{kl}/n\rightarrow m_{kl}\in [0,1],$
where the functions $f_{kl}$ satisfy the same assumptions as $f_k$ in (T1) and (T2).
\end{itemize}
Let us shortly comment the above assumptions. Condition (T1) makes the explicit computation of the asymptotic covariance
matrix in the forthcoming central limit theorem possible. Condition (T3) implies that the observation numbers $n_k$ have the same order.
Condition (T2) means that the points of the $l$th grid
do not lie dense between any two successive points of the $k$th grid, i.e.
the number of points $t_j^l$ that lie in the interval $[t_{i-1}^k, t_i^k]$ is uniformly bounded by
a constant
for all $1\leq k,l\leq d$
(cf. Lemma \ref{prooflem1} for a closely related result).
 When these last two conditions
(similar number of observations and uniform boundedness of the number of points $t_j^l$ that belong to $[t_{i-1}^k, t_i^k]$)
are fulfilled
we say that the sampling schemes are
{\it comparable}. Finally, condition (T4) means that the number of common points can be negligible compared to $n$ (if $m_{kl}=0$) or it
can be of order $n$ (if $m_{kl}>0$).

We want to emphasize that the full force of Assumption (T) is only required for the proof of the central limit theorem! For the
consistency result and the rate of convergence it suffices to assume that the grids $(t_i^k)$, $k=1, \ldots, d,$ are comparable.
In particular, the representation
(\ref{ttrans}) and the condition (T4) are not required.

Now we explain the construction of our estimator $HY^n$.
First, we choose a window size $k_n$ as
\begin{equation} \label{kn}
k_n=\theta \sqrt{n} +o(n^{1/4})
\end{equation}
for some constant $\theta >0$. In the next step we choose a positive weight function $g:[0,1]\rightarrow \R$
with $g(0)=g(1)=0$, which is piecewise $C^1$ with piecewise Lipschitz derivative $g'$ and $\int_0^1 g^2(x)dx>0$.
For any $d$-dimensional  stochastic process $V=(V^1, \ldots, V^d)$
we define  the quantity
\begin{equation} \label{preave}
\overline V_{t_i^k}^k = \sum_{j=1}^{k_n-1} g \Big(\frac{j}{k_n} \Big) \De_{t_{i+j}^k} V^k, \qquad
\De_{t_{i+j}^k} V^k =  V_{t_{i+j}^k}^k - V_{t_{i+j-1}^k}^k,
\end{equation}
which we call {\it pre-averaging in tick time}. The name refers to the fact that we use the same amount
of data to construct $\overline V_{t_i^k}^k$ for all $1\leq k\leq d$; alternatively one could perform the
{\it pre-averaging in calendar time} by using the same time interval for all coordinates $V^k$, but with different number
of observations in each time window. The latter approach would result in different properties of the estimator.

As discussed in \cite{JLMPV}, \cite{JPV} or \cite{PV2} the local averages technique performed in (\ref{preave}) diminishes
the influence of the noise process $\varepsilon$ to some extent (but not completely) and helps us to get information
about $\Sigma$. In the last step, as proposed in \cite{CKP}, we define a Hayashi-Yoshida type estimator based on pre-averaged
observations by
\bea
\label{HY}
HY^n_{kl} = \frac{1}{ \left( \psi k_{n} \right)^{2}} \sum_{i = 0}^{n_{k} - k_{n} + 1} \sum_{j = 0}^{n_{l} - k_{n} + 1}
\overline Y_{t_i^k}^k\overline Y_{t_j^l}^l 1_{ \{ (t_{i}^{k}, t_{i + k_n}^{k}] \cap (t_j^{l}, t_{j+k_n}^{l}] \neq \emptyset \}}
\eea
with $\psi=\int_0^1 g(x) dx$, and set $HY^n=(HY^n_{kl})_{1\leq k,l\leq d}$. In \cite{CKP} we have already indicated the consistency
of $HY^n$. The aim of this paper is to provide the complete asymptotic theory to be able to construct
confidence regions for the quadratic covariation $[X]$.

\section{The asymptotic theory} \label{sec2}
\setcounter{equation}{0}
\renewcommand{\theequation}{\thesection.\arabic{equation}}

We start with the consistency of the estimator $HY^n$ which has been shown in \cite{CKP}.

\begin{theo} \label{th1}
Assume that Assumption (T) holds and that the marginal law $Q$ of $\varepsilon$ has finite fourth moments. Then we have
\bean
HY^n \toop [X]=\int_0^1 \Sigma_s ds.
\eean
\end{theo}
As we remarked above the full force of Assumption (T) is not required for the proof of Theorem \ref{th1}; it is
just the comparability of sampling times which matters (see \cite{CKP} for more details). Two remarks are in order.

\begin{rem} \label{rem1} (Univariate case) \rm
Even though no synchronization is necessary in the one-dimensional case, our estimator $HY^n$ is for $d=d^{\prime}=1$ not identical to the univariate pre-averaged estimator proposed in \cite{JLMPV}! Recall that the latter is defined as
\bean
C^n = \frac{1}{  k_{n} } \sum_{i=1}^{n-k_n+1} |\overline Y_{t_i}|^2 \toop [X]\int_0^1 g^2(x) dx  + \theta^{-2}\Psi
\int_0^1 (g'(x))^2 dx,
\eean
where we set $t_i=t_i^1$. This should be compared to the univariate version of $HY^n$, which is
\begin{equation*}
HY^n = \frac{1}{ \left( \psi k_{n} \right)^{2}} \sum_{i = k_n}^{n -2 k_{n} + 1}
\overline Y_{t_i} \Big( \sum_{j = -k_n+1}^{k_{n} - 1} \overline Y_{t_{i+j}} \Big)
\end{equation*}
plus some border terms of small order.
We see immediately that the first estimator $C^n$ is biased (even after rescaling), where the bias is coming  from
$\Psi=\E[\varepsilon_t^2]$, while our estimator $HY^n$ is unbiased. The reason for this is the additional averaging
performed by $HY^n$ (which is taken care by the second sum in the above formula).
Indeed, the factor in front of $\varepsilon_{t_i}^2$ for $\frac{k_n}{n}\leq i\leq 1- \frac{k_n}{n}$
is equal to
\begin{equation*}
\left( \sum_{j=0}^{k_n-1} g \Big(\frac{j+1}{k_n} \Big) -g \Big(\frac{j}{k_n} \Big) \right)^2 = (g(1)-g(0))^2=0,
\end{equation*}
which explains why $\Psi$ does not appear in the limit of $HY^n$. The unbiasedness of $HY^n$ is an important feature
as the estimation of the covariance matrix $\Psi$ of the noise can be problematic in practice, because
we strongly rely on the iid assumption on the noise process to successfully perform the estimation of $\Psi$.
Let us remark that pre-averaging in calendar time would
also lead to a bias.
\end{rem}

\begin{rem} \label{rem2} (m-dependent noise) \rm
Let us study the case of an $m$-dependent noise process. More precisely, we consider the multivariate discrete model
$Y_{t_i^k}^k=X_{t_i^k}^k + \varepsilon_{t_i^k}^k$, $k=1, \ldots, d$, $i=0,\ldots, n_k,$
where all previous assumptions are satisfied except the noise process is now {\it m-dependent in tick time},
which means that for $t_i^k\leq t_j^l$  the random variables
$\varepsilon_{t_i^k}^k$ and $\varepsilon_{t_j^l}^l$ are independent, if $\|t_i^k - t_j^l\|>m$ with
\begin{equation*}
\|t_i^k - t_j^l\| = \min (j-\max\{z|~ t_z^l\leq t_i^k \}, \min\{z|~ t_z^k\geq t_j^l \}-i ),
\end{equation*}
and similarly for $t_j^l < t_i^k$. These types of models are important from the practical point of view. Our previous iid assumption
on the noise process implies that $\varepsilon_{t_i^k}^k$ and $\varepsilon_{t_j^l}^l$ are possibly correlated
when $t_i^k=t_j^l$; on the other hand they are independent even when the grid points $t_i^k$ and $t_j^l$ lie
arbitrarily close, say less than a second apart. Such an assumption might be not very plausible from the finance point of view.

In the case of $m$-dependent noise the estimator $HY^n$ still remains consistent, i.e. $HY^n$ is robust to $m$-dependence in
tick time.
As in the previous remark only the products $\varepsilon_{t_i^k}^k \varepsilon_{t_j^l}^l$ with
$\|t_i^k - t_j^l\|\leq m$ play a role when computing the bias. But these terms have asymptotically the same weight
as for instance $(\varepsilon_{t_i^k}^k)^2$, which is $0$ (see Remark \ref{rem1}). Thus, $HY^n$ is unbiased.
\end{rem}
In order to describe the weak limit associated with $HY^n - [X]$ we need to introduce various notations. \\ \\
{\it Notation.} Let us first extend the weight function $g$ to the whole real line by setting $g(x)=0$
for $x \not \in [0,1]$. We set for $x\in [0,1]$
\begin{equation} \label{hfun}
h_{kl} (x)= \frac{m_k f'_k(x)}{m_l f'_l(x)}, \qquad 1\leq k,l\leq d,
\end{equation}
where $f_k$ resp. $m_k$ are given  in (\ref{ttrans}) resp. (\ref{nk}). Now we define two sets of functions, namely
\begin{equation} \label{psifun}
\left.\begin{array}{l}
\psi (s,x)=  \int_0^1 \int_{(u-1+s)x}^{1+x(s+u)} g(u) g(v) dv du, \\[4.5 ex]
\overline \psi (s,x)=  \int_0^1 \int_{(u-1+s)x}^{1+x(s+u)} g(u) g'(v) dv du, \\[4.5 ex]
\widetilde \psi (s,x)=  \int_0^1 \int_{(u-1+s)x}^{1+x(s+u)} g'(u) g'(v) dv du,
\end{array}\right\}
\end{equation}
and
\begin{equation} \label{gammafun}
\left.\begin{array}{l}
\gamma_{kl,k'l'}(u) = \frac{1}{m_l f'_l(u)} \int_{-(1+h_{lk}(u))}^{1+h_{lk}(u)}  \psi (s, h_{kl}(u))
\psi (h_{l'l}(u)s, h_{k'l'}(u)) ds, \\[4.5 ex]
\overline \gamma_{kl,k'l'}(u) =  \frac{m_{kk'} f'_{kk'}(u)}{m_l f'_l(u)} \int_{-(1+h_{lk}(u))}^{1+h_{lk}(u)}
\overline \psi (s, h_{kl}(u))
\overline \psi (h_{l'l}(u)s, h_{k'l'}(u)) ds, \\[4.5 ex]
\widetilde \gamma_{kl,k'l'}(u) =  \frac{m_{kk'} f'_{kk'}(u) m_{ll'} f'_{ll'}(u)}{m_l f'_l(u)} \int_{-(1+h_{lk}(u))}^{1+h_{lk}(u)}  \widetilde \psi (s, h_{kl}(u))
\widetilde  \psi (h_{l'l}(u)s, h_{k'l'}(u)) ds,
\end{array}\right\}
\end{equation}
for $s\in \R$, $1\leq k,k',l,l'\leq d$ and $u\in [0,1]$. Notice that when for example the number of joint points between the $k$th
and $k'$th grid is negligible compared to $n$ (which can only hold for $k \not=k'$) then $m_{kk'}=0$. In this case we have
$\overline \gamma_{kl,k'l'} \equiv \widetilde \gamma_{kl,k'l'} \equiv0$.\\ \\
Before we present the stable central limit theorem let us recall the notion of stable convergence. A sequence
of random variables $Z^n$ on $(\Omega, \mathcal F, \PP)$ converges stably in law towards
 $Z$, written $Z_n \stab Z$, with $Z$ being defined on an extension
$(\Omega', \mathcal F', \PP')$ of the original probability space $(\Omega, \mathcal F, \PP)$,
iff for any bounded, continuous real-valued  function $g$ and any bounded $\mathcal{F}$-measurable random variable $V$ it holds that
$\E[ g(Z_n) V] \rightarrow \E'[ g(Z) V]$
as $n\rightarrow \infty$. We refer to \cite{AE}, \cite{REN} or \cite{JS} for more details on stable convergence. The next theorem
is the main result of our paper, and its proof is postponed to Section \ref{sec4}.

\begin{theo} \label{th2}
Assume that Assumption (T) holds and that the marginal law $Q$ of $\varepsilon$ has finite eighth moments. Then the
sequence $L^n = n^{1/4}(HY^n - [X])$ converges stably in law towards a random variable $L$, defined on an extension
$(\Omega', \mathcal F', \PP')$ of the original probability space $(\Omega, \mathcal F, \PP)$, and $L$ has a centered mixed normal
distribution, i.e. conditionally on $\mathcal F$, $L=(L_{kl})_{1\leq k,l\leq d}$ has a centered normal distribution with
\bean
\E'[L_{kl} L_{k'l'}|\mathcal F] = V_{kl,k'l'}, \qquad  1\leq k,k',l,l'\leq d,
\eean
where the random variable $V_{kl,k'l'}$ is defined via
\begin{eqnarray} \label{variance}
&& V_{kl,k'l'} = \frac{1}{\psi^4} \int_0^1 \Big\{ \theta
\Big(\gamma_{kl,k'l'}(u) \Sigma_u^{kk'} \Sigma_u^{ll'} + \gamma_{kl,l'k'}(u) \Sigma_u^{kl'}\Sigma_u^{lk'}\Big)  \nonumber \\[2.0 ex]
&&+
\theta^{-1}  \Big( \Psi^{ll'} \overline \gamma_{lk,l'k'}(u) \Sigma_u^{kk'}
+ \Psi^{lk'} \overline \gamma_{lk,k'l'}(u) \Sigma_u^{kl'}
+ \Psi^{kl'} \overline \gamma_{kl,l'k'}(u) \Sigma_u^{lk'}
+ \Psi^{kk'} \overline \gamma_{kl,k'l'}(u) \Sigma_u^{ll'} \Big) \nonumber \\[2.0 ex]
&&+ \theta^{-3} \Big(\Psi^{kk'} \Psi^{ll'} \widetilde \gamma_{kl,k'l'}(u) + \Psi^{kl'} \Psi^{lk'} \widetilde \gamma_{kl,l'k'}(u) \Big) \Big \} du,
\end{eqnarray}
and the functions $ \gamma_{kl,k'l'} ,  \overline \gamma_{kl,k'l'}, \widetilde  \gamma_{kl,k'l'}$
are given by (\ref{gammafun}) and $\theta$ is defined in (\ref{kn}). We also write $L\sim MN(0,V)$ to denote the centered mixed normal distribution with
random $\mathcal F$-measurable covariance matrix $V=(V_{kl,k'l'})_{1\leq k,k',l,l'\leq d}$ above.
\end{theo}

The rate of convergence $n^{-1/4}$ is known to be optimal for the parametric analogue of our
estimation problem (i.e. when the process $\Sigma$ is constant); see e.g. \cite{B} or \cite{GJ}.
We remark that the covariance matrix $\Psi$ of the noise process $\varepsilon$ always appears in the representation of $V$
as $\overline \gamma_{kk,kk}(u), \widetilde \gamma_{kk,kk}(u)>0$ for all $1\leq k\leq d$.

\begin{rem} \label{rem3} (Univariate case) \rm
In the one-dimensional case ($d=d^{\prime}=1$) we deduce that
\begin{equation*}
n^{1/4}\Big(HY^n - \int_0^1 \sigma_s^2 ds \Big) \stab MN(0,V),
\end{equation*}
where the expression for $V$ simplifies to
\begin{eqnarray} \label{varianceone}
V = \frac{2}{\psi^4} \left( \theta \kappa
\int_0^1 \frac{\sigma_u^4}{f'(u)} du   +
2\theta^{-1} \Psi  \overline \kappa \int_0^1 \sigma_u^2 du + \theta^{-3} \Psi^2 \widetilde \kappa \right)
\end{eqnarray}
with
\bea \label{gammafunone}
\kappa = \int_{-2}^{2}  \psi^2 (s, 1) ds , \quad
\overline \kappa =  \int_{-2}^{2}
\overline \psi^2 (s, 1)
ds , \quad
\widetilde \kappa =  \int_{-2}^{2}
\widetilde \psi^2 (s,1)
ds.
\eea Note that we have $f_{11}=f_{1}=:f$, $h_{11}=1$ and $m_{11}=m_1=1$, as well as $\int_0^1 f'(u) du =1$. If we further deal with equidistant data
it follows that $f(u)=u$.

To measure the quality of $HY^n$ compared to alternative estimators in the one-dimensional setting, it is common to compute $V$ in the parametric model of zero drift and a constant volatility $\sigma$. In case of equidistant observations we know from \cite{GJ} that the lower bound for the variance is then given by $8 \sigma^3 \sqrt{\Psi}$. If we choose the (probably) simplest weight function given by $g(x) =\min(x,1-x)$, some lengthy calculations give
\bean
\kappa = \frac{7585}{1161216}, \quad \overline \kappa = \frac{151}{20160}, \quad \widetilde \kappa = \frac{1}{24}, \quad \psi = \frac 14,
\eean
and the optimal choice of $\theta$ corresponds to $\theta^\star \approx 2.381 {\sqrt{\Psi}}/{\sigma}$. Overall we obtain a minimal variance of $12.765 \sigma^3 \sqrt{\Psi}$. This is quite close to the efficiency bound and also to the minimal variance of (the bias corrected version of) $C^n$, the original pre-averaged statistic for equidistant data from \cite{JLMPV}, which is about $8.545 \sigma^3 \sqrt{\Psi}$. This mild loss in efficiency is the price we have to pay for the additional robustness property discussed in Remark \ref{rem2}.
\end{rem}

\section{Estimation of variance} \label{varest}
\setcounter{equation}{0}
\renewcommand{\theequation}{\thesection.\arabic{equation}}

To transform the probabilistic result of Theorem \ref{th2} into a feasible statistical one, we need to find a consistent estimator
of the conditional covariance matrix $V$ defined by (\ref{variance}). We will introduce three different approaches to solve this task -- a general one, which works in arbitrary dimensions and does not require information of the time transforming functions; a second estimator, which uses local estimates of the volatility $\Sigma $; a third one tuned for the one-dimensional case, where the variance becomes particularly simple as seen in Remark \ref{rem3}. All proofs are given in Section \ref{sec4}.

Let us begin with the first estimator, for which we benefit from related work in \cite{M}, where an estimator for the variance of the usual Hayashi-Yoshida estimator in the no-noise case was constructed. We introduce a second auxiliary sequence $\beta_n = \varpi n^{\eta} + o(n^{\eta})$, $\varpi > 0, \eta \in (0,1)$, and compute for each $\alpha \in \{0, \ldots [n/\beta_n]-1 \}$ the statistic
\bee
HY^n_{kl}(\alpha) = \frac{1}{ \left( \psi k_{n} \right)^{2}} \sum_{t_i^k \in B_n(\alpha)} \sum_{j = 0}^{n_{l} - k_{n} + 1}
\overline Y_{t_i^k}^k\overline Y_{t_j^l}^l 1_{ \{ (t_{i}^{k}, t_{i + k_n}^{k}] \cap (t_j^{l}, t_{j+k_n}^{l}] \neq \emptyset \}},
\eee
which is essentially the same quantity as $HY^n_{kl}$, but we only sum over time points $t_i^k$ from the smaller interval $B_n(\al) = [\frac{\al \be_n}n,\frac{(\al+1)\be_n}n)$. We set
\bee \label{vari1}
V^{n,1}_{kl,k'l'} = \frac{ \sqrt n}{2} \sum_{\alpha=1}^{[\frac{n}{\beta_n}]-1} \Big( 2 HY^n_{kl}(\al) HY^n_{k'l'}(\al) - HY^n_{kl}(\al) HY^n_{k'l'}(\al-1) - HY^n_{kl}(\al-1) HY^n_{k'l'}(\al) \Big).
\eee
This estimator is based on a local estimation of the covariance of $HY^n_{kl}$ and $HY^n_{k'l'}$. In order to obtain reasonable estimates for this covariance on the interval $B_n(\al)$, we use $HY^n_{kl}(\al) HY^n_{k'l'}(\al)$ to mimic the covariance of interest plus the product of the expectations of both factors. The latter bias is corrected by quantities like $HY^n_{kl}(\al) HY^n_{k'l'}(\al-1)$, where we use the usual ``conditional independence'' of increments of $Y$ over disjoint intervals. $V^{n,1}_{kl,k'l'}$ is now constructed as a symmetrized version of these local estimates, and we sum up over all $a$ afterwards to obtain a global one.

A drawback of this construction is that we need an additional condition on the process $\sigma$. In order for $HY^n_{kl}(\al)$ and $HY^n_{kl}(\al-1)$ to estimate the same quantity up to an error small enough, one usually postulates that $\sigma$ is an It\^o semimartingale itself. Under a furher assumption on $\eta$ we have the following theorem.

\begin{theo} \label{th3}
Assume that Assumption (T) holds and that the marginal law $Q$ of $\varepsilon$ has finite eighth moments. Furthermore, suppose that $\sigma$ is a $d\times d^{\prime}$-semimartingale of the form (\ref{xpr}) as well and let $1/2 < \eta < 2/3$. Then we have $V^{n,1}_{kl,k'l'} \toop V_{kl,k'l'}.$
\end{theo}

As mentioned above, the second estimator uses local estimates of the volatility $\Sigma$ and the covariance matrix $\Psi$ of the noise, and we assume knowledge of the time-transforming functions $f_k$ and $f_{kl}$, which in practice have to be approximated via the observed time points.

We start with the construction of the estimator of $\Sigma_s$. We define $HY^n ([0,t])=(HY^n_{kl} ([0,t]))_{1\leq k,l\leq d}$
for $t\in [0,1]$ by
\bean
HY^n_{kl}([0,t]) = \frac{1}{ \left( \psi k_{n} \right)^{2}} \sum_{i:~  t_{i + k_n}^{k}\leq t} \sum_{j:~t_{j+k_n}^{l}\leq t}
\overline Y_{t_i^k}^k\overline Y_{t_j^l}^l 1_{ \{ (t_{i}^{k}, t_{i + k_n}^{k}] \cap (t_j^{l}, t_{j+k_n}^{l}] \neq \emptyset \}}
\eean
which is consistent for the integrated covariation matrix up to time $t$. As the volatility process $(\Sigma_s)_{s\in [0,1]}$ is left-continuous, it is a natural idea to estimate $\Sigma_s$ via
\bean
\Sigma_{s,n}= \frac{HY^n([0,s]) - HY^n([0,s-l_n])}{l_n}
\eean
for some sequence $l_n$ with $l_n\rightarrow 0$, $\sqrt{n} l_n \rightarrow \infty$ and $s\in [l_n,1]$ (for $s\in [0,l_n]$ we set
$\Sigma_{s,n}=\Sigma_{l_n,n}$). The condition $\sqrt{n} l_n \rightarrow \infty$ is required
to guarantee a sufficient amount of asymptotically uncorrelated summands in the definition of $\Sigma_{s,n}$.

The estimation of the covariance matrix $\Psi$ is somewhat easier. Recall that $(t_p^{kl})_{1\leq p\leq n_{kl}}$ denotes
the set of common points of the $k$th and the $l$th grid, and define $i(p,k,l)=i$ with $t_i^k= t_p^{kl}$
for arbitrary $k,l=1, \ldots, d$. The estimator of $\Psi^{kl}$ is now given as
\begin{equation} \label{psiest}
\Psi_{n}^{kl}= -\frac{1}{n_{kl}}\sum_{p=1}^{n_{kl}} \Delta_{t_{i(p,k,l)}^k} Y^k \Delta_{t_{i(p,l,k)+1}^l} Y^l.
\end{equation}
The intuition behind this estimator is rather simple. First of all, since the increments of $X$ at highest frequency
converge to $0$ almost surely, the process $Y$ can be replaced by $\varepsilon$ without any changes in the limit. For this reason
 the estimator $\Psi_{n}^{kl}$ converges to $\Psi^{kl}$ almost surely by the strong law of large numbers (applied to the iid process
$\varepsilon $) if $n_{kl}\rightarrow \infty$.  When the sequence $n_{kl}$ does not diverge to $\infty$ then the convergence
does not hold, but we have $n_{kl}/n\rightarrow m_{kl}=0$. Thus the corresponding functions $\overline \gamma$ and $\widetilde \gamma$ vanish as well, and this will be sufficient for the estimation of $V$. \\ \\
After all we obtain the following result.

\begin{theo} \label{th4}
Assume that Assumption (T) holds and that the marginal law $Q$ of $\varepsilon$ has finite eighth moments. Then we have
\bean
&& V^{n,2}_{kl,k'l'} := \frac{1}{\psi^4} \int_0^1 \Big\{ \theta
\Big(\gamma_{kl,k'l'}(u) \Sigma_{u,n}^{kk'} \Sigma_{u,n}^{ll'} +
\gamma_{kl,l'k'}(u) \Sigma_{u,n}^{kl'}\Sigma_{u,n}^{lk'}\Big)  \nonumber \\[2.0 ex]
&&+
\theta^{-1}  \Big( \Psi^{ll'}_n \overline \gamma_{lk,l'k'}(u) \Sigma_{u,n}^{kk'}
+ \Psi^{lk'}_n \overline \gamma_{lk,k'l'}(u) \Sigma_{u,n}^{kl'}
+ \Psi^{kl'}_n \overline \gamma_{kl,l'k'}(u) \Sigma_{u,n}^{lk'}
+ \Psi^{kk'}_n \overline \gamma_{kl,k'l'}(u) \Sigma_{u,n}^{ll'} \Big) \nonumber \\[2.0 ex]
&&+ \theta^{-3} \Big(\Psi^{kk'}_n \Psi^{ll'}_n \widetilde \gamma_{kl,k'l'}(u) + \Psi^{kl'}_n \Psi^{lk'}_n
 \widetilde \gamma_{kl,l'k'}(u) \Big) \Big \} du \toop V_{kl,k'l'}.
\eean
\end{theo}

Let us finally focus on the one-dimensional case and recall the asymptotic variance in (\ref{varianceone}). As noted before, we do not have to care about any of the $\kappa$'s from (\ref{gammafunone}), as they can directly be computed from our choice of $g$. Using the univariate version of the estimator in $(\ref{psiest})$ for $\Psi$ (which is consistent now) and the Hayashi-Yoshida type estimator $HY^n$ for $\int_0^1 \si_u^2 du$, all we need to find is a feasible estimator for the rescaled integrated quarticity $\int_0^1 \frac{\sigma_u^4}{f'(u)} du $. Among several possibilities (including yet another Hayashi-Yoshida type one) we have decided to go with a pre-averaged version of realized quarticity. Thus we set
\bea
\mu = \int_0^1 g^2(u) du, \quad \widetilde \mu = \int_0^1 (g')^2(u) du,
\eea
and define
\bea
V^{n,3} = \frac{2}{\psi^4} \left(\frac{\kappa }{3 \theta \mu^2} \sum_{i=1}^{n-k_n+1} |\overline Y_{t_i}|^4 + \frac{2}{\theta} \Psi_{n} HY^n \Big(\overline \kappa - \frac{\kappa \widetilde \mu}{\mu} \Big) + \frac{1}{\theta^3} \Psi^2_{n} \Big(\widetilde \kappa - \frac{\kappa \widetilde \mu^2}{\mu^2}\Big) \right).
\eea
The result precisely reads as follows.

\begin{theo} \label{th5}
Let $d=1$ and assume that Assumption (T) holds and that the marginal law $Q$ of $\varepsilon$ has finite eighth moments. Then we have $V^{n,3} \toop V.$
\end{theo}

In order to present a feasible central limit theorem associated with Theorem \ref{th2} we vectorize the quantities $HY^n$
and $[X]$, i.e.
\begin{equation*}
\widehat{HY}^n = \mbox{vec}(HY^n), \qquad \widehat{[X]} = \mbox{vec}([X]),
\end{equation*}
where vec is the vectorization operator that stacks columns of a matrix below one another, and set
\begin{eqnarray*}
\widehat{V}_{kl} &=& V_{k-d[(k-1)/d], [(k-1)/d]+1, l-d[(l-1)/d], [(l-1)/d]+1}, \\[1.5 ex]
\widehat{V}_{kl}^{n,b} &=& V_{k-d[(k-1)/d], [(k-1)/d]+1, l-d[(l-1)/d], [(l-1)/d]+1}^{n,b}
\end{eqnarray*}
with $1\leq k,l\leq d^2$ and $b= 1,2,3$.
Now, the properties of stable convergence imply the following result, which can be directly applied for the construction
of confidence regions.

\begin{cor} \label{cor1}
Under the assumptions of Theorem \ref{th2} we obtain the stable convergence
\begin{equation*}
n^{1/4}(\widehat{HY}^n  - \widehat{[X]}) \stab MN(0, \widehat{V}).
\end{equation*}
Also, for any $b=1,2,3$ and as long as the conditions for the corresponding theorem above are satisfied, we have the standard central limit theorem
\begin{equation} \label{standclt}
n^{1/4}(\widehat{V}^{n,b})^{-1/2}(\widehat{HY}^n  - \widehat{[X]}) \schw N_{d^2}(0, I_{d^2}),
\end{equation}
where $N_{d^2}(0, I_{d^2})$ denotes the $d^2$-dimensional normal  distribution with covariance matrix equal to identity,
and $\widehat{V}=(\widehat{V}_{kl})_{1\leq k,l\leq d^2}$, $\widehat{V}^{n,b}=(\widehat{V}_{kl}^{n,b})_{1\leq k,l\leq d^2}$.
\end{cor}

\begin{rem} \label{rem4} (m-dependent noise) \rm We have indicated in Remark \ref{rem2} that the consistency result for the Hayashi-Yoshida type estimator $HY^n$ from Theorem \ref{th1} remains valid, if the assumption of independent noise variables is weakened to $m$-dependence. This does obviously not hold for the central limit theorem, as the particular form of the noise part of the asymptotic variance relies heavily on the independence assumption. Nevertheless, even in this framework a central limit theorem can be shown, but for the sake of brevity we dispense with the specification of its precise form. It is worth noticing, however, that $V^{n,1}_{kl,k'l'}$ by construction remains a consistent estimator for the asymptotic variance in this rather general setting, as it is designed to mimic the covariance of $HY^n_{kl}$ and $HY^n_{k'l'}$ without using any prior knowledge on $\varepsilon$ apart from dependence on only a finite number of neighbours. Therefore Theorem \ref{th3} and thus in turn (\ref{standclt}) for $b=1$
hold true for $m$-dependent noise as well.
\end{rem}

\section{Numerical study} \label{sec3}
\setcounter{equation}{0}
\renewcommand{\theequation}{\thesection.\arabic{equation}}

Here, we supplement the above asymptotic results based on $n \to \infty$ with a finite sample analysis by using Monte Carlo experiments. We simulate a bivariate stochastic volatility model with noise, as was also conducted in previous work of \cite{BHLS2} and \cite{CKP}.

More specifically, to simulate efficient log-prices we consider
\begin{equation}
\text{d}X_{t}^{(i)} = a^{(i)} \text{d}t + \rho^{(i)} \sigma_{t}^{(i)} \text{d}B_{t}^{(i)} + \sqrt{1 - [\rho^{(i)}]^2} \sigma_{t}^{(i)} \text{d}W_{t},
\end{equation}
where $B^{(i)} \Perp W$. Throughout, we work with $i = 1,2$. Note that $\rho^{(i)} \sigma_{t}^{(i)} \text{d}B_{t}^{(i)}$ represents an idiosyncratic shock, while $\sqrt{1 - [\rho^{(i)}]^2} \sigma_{t}^{(i)} \text{d}W_{t}$ is a common factor.

The model for the diffusive volatility is specified as: $\sigma_{t}^{(i)} = \exp (\beta_0^{(i)} + \beta_1^{(i)} \varrho_{t}^{(i)})$, where each of the $\varrho_{t}^{(i)}$ processes conform with Ornstein-Uhlenbeck dynamics: $\text{d} \varrho_{t}^{(i)} = \alpha^{(i)} \varrho_{t}^{(i)} \text{d}t + \text{d}B_{t}^{(i)}$. This assumption means that the innovations of $\rho^{(i)} \sigma_{t}^{(i)} \text{d}B_{t}^{(i)}$ and $\text{d} \sigma_{t}^{(i)}$ are perfectly correlated, while the covariation between $\text{d}X_{t}^{(i)}$ and $\text{d}\varrho_{t}^{(i)}$ is equal to $\rho^{(i)} \sigma_{t}^{(i)} \text{d}t$.
Finally, note that the model allows the two underlying price processes $X_{t}^{(1)}$ and $X_{t}^{(2)}$ to be correlated in the magnitude of $\sqrt{1 - [\rho^{(1)}]^2} \sqrt{1 - [\rho^{(2)}]^2}$.

We carry out our numerical experiments by using the following parametrization, assumed to be identical across the two volatility factors: $(a^{(i)}, \beta_0^{(i)}, \beta_1^{(i)}, \alpha^{(i)}, \rho^{(i)}) = (0.03, -5/16, 1/8,-1/40,-0.3)$, so that $\beta_0^{(i)} = [\beta_1^{(i)}]^{2} / [2 \alpha^{(i)}]$. This choice of parameters implies that integrated volatility has been normalized, in the sense that $\mathbb{E} \Bigl( \int_{0}^{1} [\sigma_s^{(i)}]^2 \text{d}s \Bigr) = 1$.

We simulate 10,000 paths of this model over the interval $[0,1]$, which we partition into $N = 23,400$ subintervals of equal length $1 / N$. In constructing noisy prices $Y^{(i)}$, we first generate a complete high-frequency record of $N$ equidistant observations of the efficient price $X^{(i)}$ using a standard Euler scheme.\footnote{Note that the Ornstein-Uhlenbeck process permits an exact discretization (see, e.g., \cite{GM}). We use that fact here to avoid committing errors in working out the discrete time distribution of $\text{d} \varrho^{(i)}$ over time steps of size $1 / N$.} The initial values for the $\varrho_{t}^{(i)}$ processes at each simulation run are drawn randomly from their stationary distribution, which is $\varrho_{t}^{(i)} \sim N(0,[-2\alpha^{(i)}]^{-1})$.

Next, we add simulated microstructure noise $Y^{(i)} = X^{(i)} + \varepsilon^{(i)}$ by taking
\begin{equation}
\varepsilon^{(i)} \mid \{ \sigma, X \} \ \overset{\text{i.i.d}}{\sim} \ N(0, \omega^2) \quad \text{with} \quad \omega^{2} = \gamma^2 \left( \frac{1}{N} \sum_{j=1}^N \sigma_{j / N}^{(i)2} \right),
\end{equation}
where $\gamma$ is the so-called noise ratio parameter. This choice means that the variance of the noise process increases with the level of volatility of $X^{(i)}$, as documented by \cite{BR}. $\gamma$ takes the value 0.50, which is a typical level of noise (e.g., \cite{COP}).

\setlength{\tabcolsep}{.25cm}
\begin{figure}[t!]
\begin{center}
\caption{Illustration of sampling schemes.
\label{Figure:SamplingScheme}}
\begin{tabular}{c}
\includegraphics[height = 6cm, width = 0.75\textwidth]{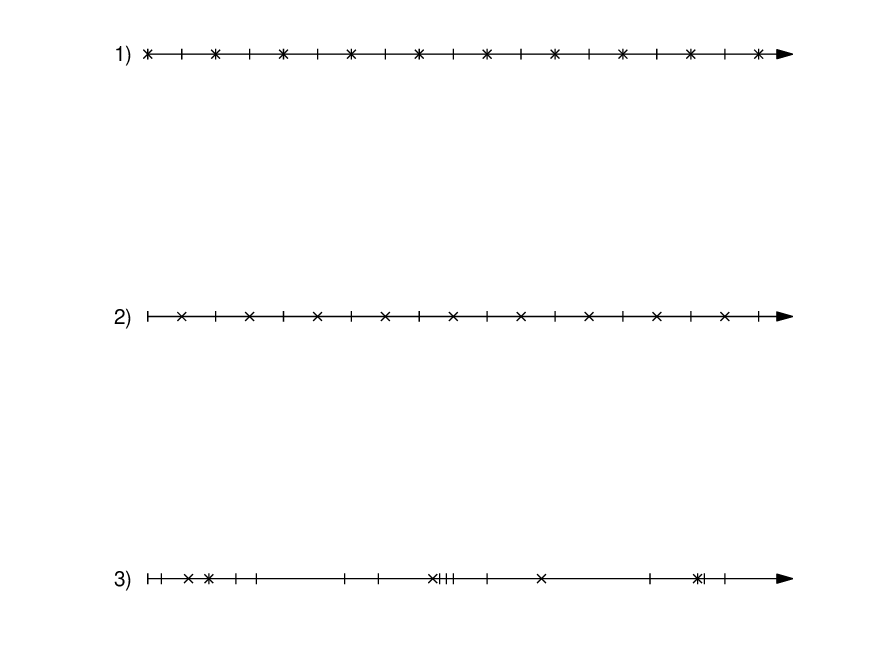}\\
\end{tabular}\smallskip
\begin{footnotesize}
\parbox{0.94\textwidth}{\emph{Note}. The figure illustrates how we design asynchronicity in the simulation study. A vertical dash (``$|$'') represents an observation of the noisy process $Y^{(1)}$, while a cross (``$\times$'') is $Y^{(2)}$ . A star (``$\ast$'') defines a common sampling point.}
\end{footnotesize}
\end{center}
\end{figure}

Finally, in order to extract non-synchronous data from the complete synchronous high-frequency record, we proceed as follows (for reference, please see Figure \ref{Figure:SamplingScheme}). We consider three settings. In scenario 1), the sampling times of $Y^{(2)}$ form a subset of the observation grid of $Y^{(1)}$, but $Y^{(1)}$ is observed more frequently. Here, we use $n_{1} = 4,680$ and $n_{2} = 2,340$. In scenario 2), we take $n_{1} = n_{2} = 4,680$, but shift the observation times of $Y^{(2)}$ to lie midway between those of $Y^{(1)}$. Finally, in scenario 3), we generate random observation times using two independent Poisson processes with intensity $\lambda_{1}$ and $\lambda_{2}$. Here $\lambda_{i}$ denotes the average waiting time for new data from process $Y^{(i)}$, so that a typical simulation will have $N / \lambda_{i}$ observations of $Y^{(i)}, i = 1,2$. We set $\lambda_{1} = 5$ and $\lambda_{2} = 10$, which implies that the first asset is trading twice as fast as the second. Note that because we are simulating in discrete time, it is possible to see common points in the last setting, as depicted in the chart.

The choice of the remaining tuning parameters are the following: We use $\theta = 0.15$ and set $k_{n} = \lceil \theta \sqrt{n} \rceil$, where $\lceil x \rceil$ is the ceil function. Moreover, to estimate the variance appearing in the CLT of $HY_{kl}^{n}$,
we use $V^{n,1}_{kl,kl}$ defined in \eqref{vari1} with $\varpi =1$ and $\eta = 7 / 12$.

Our initial numerical experimentations show that the raw estimator from Eq. \eqref{HY} is slightly downward biased in finite samples. This is familiar from related estimators, such as \cite{CKP}, where an additional factor is applied to correct for the loss of summands induced by pre-averaging. Here, the problem is slightly more delicate, but nonetheless a relatively simple device can be used to adjust the estimator. In particular, we generate a bivariate Brownian motion $(B^{(1)}, B^{(2)})$ with a known correlation $\rho$ (throughout, we use $\rho = 1$), where the coordinates of these two processes are identical to $(Y^{(1)}, Y^{(2)})$. We then estimate $R_{kl}^{n} = \mathbb{E}[HY_{kl}^{n}]$ across 10,000 repetitions using the data from $B^{(1)}$ and $B^{(2)}$ and divide the original statistic $HY_{kl}^{n}$ (based on data from $Y^{(1)}$ and $Y^{(2)}$) by $R_{kl}^{n} / \rho$. A similar procedure can be used to bias correct the estimator of variance.

\subsection{Simulation results}

In Table \ref{Table:Simulation}, we present the relative bias and root mean squared error of our pre-averaged Hayashi-Yoshida estimator. As a comparison, we also computed the modulated realised covariance (MRC) of \cite{CKP} based on refresh time sampling. As the table reveals, both estimators are unbiased (after bias correction) in all three scenarios. $HY_{22}^{n}$ does retain a slight bias in those scenarios, where $n_2$ is small, but the bias is less than half a percent. The rmse of $HY^{n}$ is larger than what we observe for the MRC, when the estimation target is a variance component; this observation is in line with the theoretical comparison of Remark \ref{rem3}.  This is particularly true for the slow-trading asset $Y^{(2)}$ in scenarios one and three. However, the rmse of $HY^{n}_{12}$ is smaller than the rmse of the MRC in all scenarios. This is explained by the fact that refresh time sampling essentially uses the slowest frequency and therefore highlights the advantages of $HY^n$.

\begin{table}[ht!]
\setlength{\tabcolsep}{0.35cm}
\begin{center}
\caption{Relative bias and root mean squared error.
\label{Table:Simulation}}
\smallskip
\begin{tabular}{lccccccc}
\hline
& \multicolumn{3}{c}{$HY^{n}$} && \multicolumn{3}{c}{\emph{MRC}} \\
\cline{2-4} \cline{6-8}
Target      & $\Sigma_{11}$ & $\Sigma_{12}$ & $\Sigma_{22}$ && $\Sigma_{11}$ & $\Sigma_{12}$ & $\Sigma_{22}$  \\
\hline
Scenario 1  &  1.00         &   1.00        &   1.00        &&   1.00        &   1.00        &   1.00   \\[-0.25cm]
            & (0.12)        &  (0.08)       &  (0.19)       &&  (0.11)       &  (0.09)       &  (0.12)  \\ \\
Scenario 2  &  1.00         &   1.00        &   1.00        &&   1.00        &   1.01        &   1.00   \\[-0.25cm]
            & (0.13)        &  (0.07)       &  (0.15)       &&  (0.11)       &  (0.09)       &  (0.12)  \\ \\
Scenario 3  &  1.00         &   1.00        &   1.00        &&   1.00        &   1.00        &   1.00   \\[-0.25cm]
            & (0.13)       &   (0.08)       &  (0.20)       &&  (0.12)       &  (0.10)       &  (0.12)  \\
\hline
\end{tabular}\smallskip
\begin{footnotesize}
\parbox{0.925\textwidth}{\emph{Note}. We report the relative bias and rmse of the estimators included in the simulation study. The bias
measure is equal to 1 for an unbiased estimator.}
\end{footnotesize}
\end{center}
\end{table}

Next, we turn to the accuracy of the asymptotic approximation, where we focus on estimation of integrated covariance, $\Sigma_{12}$. In Figure \ref{Figure:Simulation}, we plot the simulated finite sample distribution of the standardized $HY_{12}^{n}$ for the three setups considered here, where the variance of the estimator is accessed by $V^{n,1}_{12,12}$ as described above. Although the approximation is not perfect, the goodness of the fit is surprisingly good taking the relatively small sample into account. Also, the ordering is as expected with the second scenario offering the best approximation to the standard normal (where $n_{1} = n_{2} = 4,680$). Moreover, while the average number of observations is identical in scenario one and three, the randomness of the spacings in the latter setting slightly deteriorates the tracking of the standard normal.

\begin{figure}[t!]
\begin{center}
\caption{Accuracy of asymptotic approximation, estimation of $\Sigma_{12}$.
\label{Figure:Simulation}}
\begin{tabular}{c}
\includegraphics[height=8.00cm,width=0.60\linewidth]{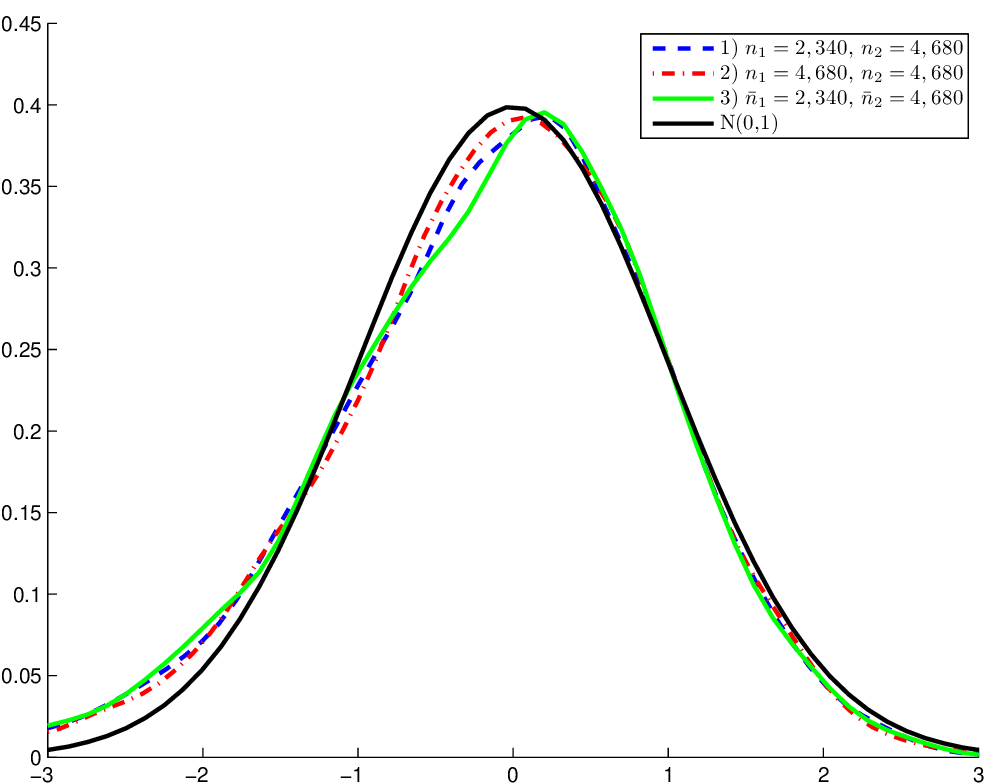}
\end{tabular}
\end{center}
\end{figure}

\section{Proofs} \label{sec4}
\setcounter{equation}{0}
\renewcommand{\theequation}{\thesection.\arabic{equation}}

Let $C>0$ denote a generic constant which may change from line to line; we also write $C_p>0$
if a constant depends on an external parameter $p$. For the sake of simplicity we will sometimes keep
the dependence of some quantities on certain parameters unreflected if things are clear from the context. Also some
notations might have a different meaning in different subsections, e.g. the quantity $R_n(p)$ stands for a generic
asymptotically negligible random variable in Sections \ref{proofth2part7}--\ref{proofth2part6}. \\ \\
We remark that all our theoretical results
(Theorems \ref{th1}, \ref{th2}, \ref{th3}, \ref{th4}, \ref{th5}) are {\it stable under localization}, i.e. if they are
valid for bounded coefficients then they remain valid for locally bounded coefficients. This means, since the processes
$a$ and $\sigma$ are c\`agl\`ad, thus locally bounded, we can assume without loss of generality:
\begin{itemize}
\item The processes $a$ and $\sigma$ are bounded in $(\omega,t)$.
\end{itemize}
See e.g. Section 3 in \cite{BGJPS} for more details.

The second important step in all proofs is the approximation
\begin{equation} \label{basicappr}
\overline Y_{t_i^k}^k \approx (\overline{\sigma_{t_i^k} W})  _{t_i^k}^k  + \overline \varepsilon _{t_i^k}^k, \qquad 1\leq k\leq d,
\end{equation}
which means that we may pretend that $a=0$ identically and that the volatility $\sigma$
is constant over the small intervals $[t_i^k, t_{i+k_n}^k]$. Indeed,
we will show that such an approximation does not affect any of our theoretical results.

Before we start proving our main results let us state some simple lemmas which concern the observation times $t_i^k$ and
the pre-averaging quantities $\overline Y_{t_i^k}^k$. In what follows we use the decomposition
\begin{equation} \label{xdec}
X_t=X_0+D_t+N_t, \qquad D_t= \int_0^t a_s ds, \qquad N_t=\int_0^t \sigma_s dW_s.
\end{equation}
We also decompose the statistic $HY^n$ as
\begin{equation} \label{hydec}
HY^n_{kl} = HY^n_{kl}[X]+HY^n_{kl}[X,\varepsilon ]+HY^n_{kl}[\varepsilon ]
\end{equation}
with
\begin{eqnarray*}
HY^n_{kl}[X] &=& \frac{1}{ \left( \psi k_{n} \right)^{2}} \sum_{i = 0}^{n_{k} - k_{n} + 1} \sum_{j = 0}^{n_{l} - k_{n} + 1}
\overline X_{t_i^k}^k\overline X_{t_j^l}^l 1_{ \{ (t_{i}^{k}, t_{i + k_n}^{k}] \cap (t_j^{l}, t_{j+k_n}^{l}] \neq \emptyset \}},
\\[1.5 ex]
HY^n_{kl}[X,\varepsilon ]  &=&
\frac{1}{ \left( \psi k_{n} \right)^{2}} \sum_{i = 0}^{n_{k} - k_{n} + 1} \sum_{j = 0}^{n_{l} - k_{n} + 1}
\Big(
\overline X_{t_i^k}^k\overline {\varepsilon }_{t_j^l}^l + \overline {\varepsilon}_{t_i^k}^k\overline X_{t_j^l}^l
\Big) 1_{ \{ (t_{i}^{k}, t_{i + k_n}^{k}] \cap (t_j^{l}, t_{j+k_n}^{l}] \neq \emptyset \}},
\\[1.5 ex]
HY^n_{kl}[\varepsilon] &=& \frac{1}{ \left( \psi k_{n} \right)^{2}} \sum_{i = 0}^{n_{k} - k_{n} + 1} \sum_{j = 0}^{n_{l} - k_{n} + 1}
\overline {\varepsilon}_{t_i^k}^k
\overline {\varepsilon }_{t_j^l}^l 1_{ \{ (t_{i}^{k}, t_{i + k_n}^{k}] \cap (t_j^{l}, t_{j+k_n}^{l}] \neq \emptyset \}}.
\end{eqnarray*}

\begin{lem} \label{prooflem1}
Under the Assumptions (T1)--(T3) we have for any $0\leq a<b\leq 1$
\begin{equation*}
\sharp \{i|~t_i^k\in [a,b]\} \leq C(b-a)n + 1 \qquad \forall 1\leq k\leq d.
\end{equation*}
\end{lem}
{\it Proof:} To compute the cardinality of the above set we need to calculate $n_k(f_k(b) -f_k(a))$, which is an upper bound for the number of points falling into $[a,b]$, up to adding one. The mean value theorem
and conditions (T2), (T3) imply that
$$n_k(f_k(b) -f_k(a)) =n_k(f_k)' (\xi ) (b-a)\leq Cn(b-a),$$
where $\xi$ is some point between $a$ and $b$. \qed \\ \\
The above lemma basically states that the amount of time points $t_i^k$ contained in $[a,b]$ is of the same order as
in the equidistant case for all $k$.

\begin{lem} \label{prooflem2}
Under the Assumptions (T) and if $\E[\varepsilon^8 ]<\infty$ we obtain for $q=2,4,8$
\begin{equation*}
\E[|\overline Y_{t_i^k}^k|^q]\leq Cn^{-q/4}, \qquad \E[|\overline D_{t_i^k}^k|^q]\leq Cn^{-q/2},
\qquad \forall 1\leq k\leq d, 1\leq i\leq n_k.
\end{equation*}
\end{lem}
{\it Proof:} These estimates are shown separately for $\overline N_{t_i^k}^k$, $\overline D_{t_i^k}^k$ and
$\overline \varepsilon _{t_i^k}^k$. They are a simple consequence of the boundedness of the processes $a$ and $\sigma$,
the Burkholder inequality and Lemma \ref{prooflem1}. See e.g.
Section 5.4 from \cite{JLMPV}  for a detailed computation in the equidistant case. \qed

\subsection{Proof of Theorem  \ref{th2}} \label{proofth2}
Because the summands in the definition of the estimator  $HY^n$ are highly correlated, the main idea of the proof
is to apply a similar method as for the proof of the central limit theorem for $m$-dependent data. Roughly speaking,
we will collect all summands of $HY^n$ in big and small blocks. The function of the small blocks is to ensure
the (conditional) asymptotic independence of the big blocks, and their contribution will become negligible in the limit.

Let us start with the formal definition of big and small blocks. For some $p>0$, we set
\begin{eqnarray} \label {bsblock}
B_z(p)&=& \Big[\frac{z(p+b)k_n}{n}, \frac{z(p+b)k_n +pk_n}{n}\Big) \qquad \mbox{(big blocks)} \nonumber \\[1.5 ex]
S_z(p)&=& \Big[ \frac{z(p+b)k_n +pk_n}{n}, \frac{(z+1)(p+b)k_n}{n}\Big) \qquad \mbox{(small blocks)}
\end{eqnarray}
where $b$ is larger than $M\max_{1\leq k\leq d} (m_k^{-1})$ and $z=0, \ldots, [\frac{n}{(p+b)k_n}]-1$. The constant $b$ is chosen
in this way to ensure that the quantities $\overline Y_{t_i^k}^k$, $\overline Y_{t_j^l}^l$ with
$t_i^k\in B_z(p)$, $t_j^l\in B_{z'}(p)$ and $z \not = z'$ do not use the same data, at least for $n$ large enough (see the proof of Lemma \ref{prooflem1}). This fact leads to the asymptotic
conditional independence of the big blocks.
The notion of big blocks comes from the fact
that the length of $B_z(p)$ is always $pk_n/n$, where we later let $p\rightarrow \infty$, which is large compared to
the length $bk_n/n$ of small blocks $S_z(p)$.

We will perform the proof in several steps. In a certain sense we will prove the statement in a reverse order.
The road map of the  proof is as follows:

\begin{itemize}
\item [(i)] In Section
\ref{proofth2part1} we will show a stable central limit theorem for the approximative quantities of the type (\ref{basicappr}),
which are collected in big blocks $B_z(p)$. The corresponding stable limit is $L$ defined in Theorem  \ref{th2}.
\item [(ii)] In Section
\ref{proofth2part2} we will prove the asymptotic negligibility of the approximative quantities of the type (\ref{basicappr})
which are collected in small blocks $S_z(p)$.
\item [(iii)] Sections \ref{proofth2part4}-\ref{proofth2part6} are devoted to the justification
of the approximation in (\ref{basicappr}): Sections \ref{proofth2part4}-\ref{proofth2part7} deal with the diffusion part
(the most involved part), Section \ref{proofth2part5}
treats the mixed part and Section \ref{proofth2part6} is devoted to the noise part.
\item [(iv)] Section \ref{proofth2part3} provides
a useful decomposition for the diffusion part, which shows that our statistic $HY^n$ is asymptotically unbiased.
\end{itemize}

\subsubsection{The central limit theorem for the big blocks} \label{proofth2part1}

Whenever $t_i^k\in A_z(p)$, $t_j^l\in A_{z'}(p)$  for $A=B$ or $A=S$ (see (\ref {bsblock})), we set
\begin{equation} \label{alpha}
\alpha_{ij}^{kl} (p) =  \frac{1}{ \left( \psi k_{n} \right)^{2}}
\Big[(\overline{\sigma_{\min A_z(p)} W})  _{t_i^k}^k  + \overline \varepsilon _{t_i^k}^k \Big]
\Big[(\overline{\sigma_{\min A_{z'}(p)} W})  _{t_j^l}^l  + \overline \varepsilon _{t_j^l}^l \Big]
1_{ \{ (t_{i}^{k}, t_{i + k_n}^{k}] \cap (t_j^{l}, t_{j+k_n}^{l}] \neq \emptyset \}}
\end{equation}
Here we follow the same approximation as in (\ref{basicappr}), except the volatility process is now frozen
in the beginning of the block $A_z(p)$ resp. $A_{z'}(p)$. We define $M_n^{kl} (p) =\sum_{z} \zeta_{zn}^{kl} (p)$
with
\bean
\zeta_{zn}^{kl} (p) = n^{1/4} \sum_{t_i^k, t_j^l\in B_z(p)} \Big(\alpha_{ij}^{kl} (p) -
\E[\alpha_{ij}^{kl} (p)| \mathcal F_{\min B_z(p)}] \Big).
\eean
As $M_n^{kl} (p)$ is a quadratic form of $Y=X+\varepsilon$, we have a straightforward decomposition
\begin{equation} \label{mndec}
M_n^{kl} (p) = M_n^{kl} (X,p) + M_n^{kl} (X,\varepsilon ,p) + M_n^{kl} (\varepsilon ,p),
\end{equation}
where $M_n^{kl} (X,p)$ denotes the diffusion part of $M_n^{kl} (p)$, $M_n^{kl} (\varepsilon ,p)$ stands for the noise
part of $M_n^{kl} (p)$ and $M_n^{kl} (X,\varepsilon ,p)$ is the mixed part of $M_n^{kl} (p)$, which will be used in the
following sections.
In these we will show that the quantities $M_n (p)$ and $L^n = n^{1/4}(HY^n - [X])$ are asymptotically equivalent, i.e.
\begin{equation} \label{neglig}
\lim_{p\rightarrow \infty } \limsup_{n\rightarrow \infty } ~\mathbb P(|M_n (p) - L^n|>\delta ) =0
\end{equation}
for all $\delta >0$.
Thus, it is sufficient to prove the following result which completes this section.
\begin{theo} \label{proofth1}
Assume that the conditions of Theorem \ref{th2} hold. Then we obtain that
\begin{equation*}
M_n (p) \stab M(p)=MN(0, V_p) \qquad \mbox{as~~} n\rightarrow \infty
\end{equation*}
for a certain conditional covariance matrix $V_p$. Furthermore, when $p\rightarrow \infty$ we deduce that $V_p \toop V$, thus
\begin{equation*}
M(p) \toop L=MN(0, V) ,
\end{equation*}
where the random variables $V$ and $L$ are defined in Theorem  \ref{th2}.
\end{theo}
\textit{Proof:} By Theorem IX.7.28 from \cite{JS} it is sufficient to show that ($1\leq k,l,k',l'\leq d$)
\begin{itemize}
\item[(i)] $\sum_{z} \E[\zeta_{zn}^{kl} (p) \zeta_{zn}^{k'l'} (p)| \mathcal F_{\min B_z(p)}] \toop V_p^{kl,k'l'},$ \\
\item[(ii)] $\sum_{z} \E[\zeta_{zn}^{kl} (p) (W_{\max B_z(p)}^{k'} - W_{\min B_z(p)}^{k'})| \mathcal F_{\min B_z(p)}] \toop 0,$ \\
\item[(iii)] $\sum_{z} \E[|\zeta_{zn}^{kl} (p)|^4]\rightarrow 0,$ \\
\item[(iv)] $\sum_{z} \E[\zeta_{zn}^{kl} (p) (N_{\max B_z(p)} - N_{\min B_z(p)})| \mathcal F_{\min B_z(p)}] \toop 0$
for all bounded martingales $N$ with $\langle N,W \rangle=0$,
\end{itemize}
to conclude the stable convergence
$M_n (p) \stab M(p)$ as $n\rightarrow \infty.$
The statement (i) is proved in the Appendix. To show (ii) we remark that the increments of $W$ involved in $\zeta_{zn}^{kl}$
are independent of $\mathcal F_{\min B_z(p)}$. On the other hand, the quantity
$\zeta_{zn}^{kl} (p) (W_{\max B_z(p)}^{k'} - W_{\min B_z(p)}^{k'})$ is an odd function of $W$ and $(W,\varepsilon) \eqschw
(-W,\varepsilon)$ since
$W,\varepsilon $ are independent, which implies that
$$\E[\zeta_{zn}^{kl} (p) (W_{\max B_z(p)}^{k'} - W_{\min B_z(p)}^{k'})| \mathcal F_{\min B_z(p)}]=0.$$
Next, to show (iii) we observe that for fixed $p$ the number of summands involved in the definition of
$\zeta_{zn}^{kl} (p)$ is $O(k_n^2)$. Due to Lemma \ref{prooflem2} and since $z=0, \ldots, [\frac{n}{(p+b)k_n}]-1$
we immediately deduce that
$$\sum_{z} \E[|\zeta_{zn}^{kl} (p)|^4]\leq C_p \frac{n}{(p+b)k_n} n k_n^8 (k_n)^{-8} n^{-2} \leq \frac{C_p}{k_n}\rightarrow 0.$$
Part (iv) is shown in \cite{JLMPV} for an analogous situation (see Lemma 5.7 therein). This completes the proof of the first statement of Theorem
\ref{proofth1}. The second statement is again proved in the Appendix. \qed

\subsubsection{Negligibility of the small blocks} \label{proofth2part2}

In this section we still consider the approximative quantities $\alpha_{ij}^{kl} (p)$ from (\ref{alpha})
and show that the term $\widetilde{M}_n^{kl} (p) =\sum_{z} \widetilde \zeta_{zn}^{kl} (p)$
with $ \widetilde \zeta_{zn}^{kl} (p) =  \sum_{i=1}^5 \widetilde \zeta_{zn}^{kl} (i,p)$ given as
\begin{eqnarray*}
\widetilde \zeta_{zn}^{kl} (1,p) &=& n^{1/4} \sum_{t_i^k, t_j^l\in S_z(p)} \Big(\alpha_{ij}^{kl} (p) -
\E[\alpha_{ij}^{kl} (p)| \mathcal F_{\min S_z(p)}] \Big) \\[1.5 ex]
\widetilde \zeta_{zn}^{kl} (2,p) &=& n^{1/4} \sum_{t_i^k\in B_{z-1}(p), t_j^l\in S_z(p)} \Big(\alpha_{ij}^{kl} (p) -
\E[\alpha_{ij}^{kl} (p)| \mathcal F_{\min B_{z-1}(p)}] \Big) \\[1.5 ex]
\widetilde \zeta_{zn}^{kl} (3,p) &=& n^{1/4} \sum_{t_i^k\in B_{z+1}(p), t_j^l\in S_z(p)} \Big(\alpha_{ij}^{kl} (p) -
\E[\alpha_{ij}^{kl} (p)| \mathcal F_{\min S_{z}(p)}] \Big)  \\[1.5 ex]
\widetilde \zeta_{zn}^{kl} (4,p) &=& n^{1/4} \sum_{t_j^l\in B_{z-1}(p), t_i^k\in S_z(p)} \Big(\alpha_{ij}^{kl} (p) -
\E[\alpha_{ij}^{kl} (p)| \mathcal F_{\min B_{z-1}(p)}] \Big) \\[1.5 ex]
\widetilde \zeta_{zn}^{kl} (5,p) &=& n^{1/4} \sum_{t_j^l\in B_{z+1}(p), t_i^k\in S_z(p)} \Big(\alpha_{ij}^{kl} (p) -
\E[\alpha_{ij}^{kl} (p)| \mathcal F_{\min S_{z}(p)}] \Big),
\end{eqnarray*}
is negligible in the sense of (\ref{neglig}). This representation holds for $p>b$ (see (\ref{bsblock})
for the definition of the constant $b$), which we assume without loss of generality. As in (\ref{mndec}), we have the decomposition
\begin{equation} \label{nndec}
\widetilde M_n^{kl} (p) = \widetilde M_n^{kl} (X,p) + \widetilde M_n^{kl} (X,\varepsilon ,p) + \widetilde M_n^{kl} (\varepsilon ,p),
\end{equation}
into the $X$-part, the mixed part and the $\varepsilon$-part, which will be used in the following sections.
Let us consider the term
$\sum_{z} \widetilde \zeta_{zn}^{kl} (1,p)$. First of all, we remark that the summands $\widetilde \zeta_{zn}^{kl} (1,p)$
are uncorrelated (as $z$ runs) and the number of summands is of order $n/(pk_n)$. Furthermore, there are $O(k_n^2)$ summands in the definition of $\widetilde \zeta_{zn}^{kl} (1,p)$.
Thus, we conclude from Lemma \ref{prooflem2} that
\begin{equation} \label{zetabound}
\E \Big( \Big| \sum_{z} \widetilde \zeta_{zn}^{kl} (1,p) \Big|^2 \Big) = \sum_{z} \E [|\widetilde \zeta_{zn}^{kl} (1,p)|^2]\leq \frac{C}{p}.
\end{equation}
Hence, we obtain
\begin{equation*}
\lim_{p\rightarrow \infty } \limsup_{n\rightarrow \infty } ~\mathbb
P \Big(\Big|\sum_{z} \widetilde \zeta_{zn}^{kl} (1,p)\Big |>\delta \Big) =0
\end{equation*}
for all $\delta >0$. The same assertion holds for $\widetilde{M}_n^{kl} (p)$, as counting the number of non-zero $\alpha_{ij}^{kl} (p)$ for $t_i^k$ and $t_j^l$ from disjoint blocks shows that the upper bound in (\ref{zetabound}) is valid for $\widetilde \zeta_{zn}^{kl} (q,p)$ as well, $q = 2, \ldots, 5$.  \qed

\subsubsection{The approximation of the diffusion part I} \label{proofth2part4}
We start with the decomposition of the diffusion part of the estimator $HY^n$. Set $
HY^n_{kl}[X] = HY^n_{kl}[D] + HY^n_{kl}[D,N] +HY^n_{kl}[N]
$
with
\begin{eqnarray*}
HY^n_{kl}[D] &=& \frac{1}{ \left( \psi k_{n} \right)^{2}} \sum_{i = 0}^{n_{k} - k_{n} + 1} \sum_{j = 0}^{n_{l} - k_{n} + 1}
\overline D_{t_i^k}^k\overline D_{t_j^l}^l 1_{ \{ (t_{i}^{k}, t_{i + k_n}^{k}] \cap (t_j^{l}, t_{j+k_n}^{l}] \neq \emptyset \}},
\\[1.5 ex]
HY^n_{kl}[D,N] &=& \frac{1}{ \left( \psi k_{n} \right)^{2}} \sum_{i = 0}^{n_{k} - k_{n} + 1} \sum_{j = 0}^{n_{l} - k_{n} + 1}
\Big(\overline D_{t_i^k}^k\overline N_{t_j^l}^l+ \overline N_{t_i^k}^k\overline D_{t_j^l}^l \Big)
1_{ \{ (t_{i}^{k}, t_{i + k_n}^{k}] \cap (t_j^{l}, t_{j+k_n}^{l}] \neq \emptyset \}},
\\[1.5 ex]
HY^n_{kl}[N] &=& \frac{1}{ \left( \psi k_{n} \right)^{2}} \sum_{i = 0}^{n_{k} - k_{n} + 1} \sum_{j = 0}^{n_{l} - k_{n} + 1}
\overline N_{t_i^k}^k\overline N_{t_j^l}^l 1_{ \{ (t_{i}^{k}, t_{i + k_n}^{k}] \cap (t_j^{l}, t_{j+k_n}^{l}] \neq \emptyset \}},
\end{eqnarray*}
where the processes $D$ and $N$ are given in (\ref{xdec}).
In this section we will show that drift part $D$ of $X$ does not influence the central limit theorem, i.e.
\begin{equation*}
HY^n_{kl}[D]=o_{\mathbb P} (n^{-1/4}), \qquad  HY^n_{kl}[D,N]=o_{\mathbb P} (n^{-1/4}).
\end{equation*}
We start with the term $HY^n_{kl}[D]$. Note that $HY^n_{kl}[D]$ contains $O(nk_n)$ non-zero summands (due to Lemma
\ref{prooflem1}). Lemma \ref{prooflem2} and the Cauchy-Schwarz inequality imply that each summand satisfies $\E[|\overline D_{t_i^k}^k\overline D_{t_j^l}^l|]\leq C n^{-1}.$ Thus, $\E[|HY^n_{kl}[D]|]\leq C n^{-1/2},$ which implies $HY^n_{kl}[D]=o_{\mathbb P} (n^{-1/4})$.

The treatment of $HY^n_{kl}[D,N]$ is a bit more delicate. We set
\begin{equation}
\xi_{ij}^n= \overline D_{t_i^k}^k\overline N_{t_j^l}^l+ \overline N_{t_i^k}^k\overline D_{t_j^l}^l
\end{equation}
and define
\begin{equation}
\widetilde \xi_{ij}^n = a_{t_i^k \wedge t_j^l}
\Big(\overline {\mbox{id}}_{t_i^k}^k\overline N_{t_j^l}^l+ \overline N_{t_i^k}^k\overline {\mbox{id}}_{t_j^l}^l \Big),
\end{equation}
where $\mbox{id}$ denotes the identity function on $\R$. The latter approximates $\xi_{ij}^n$ by freezing the process $a$
in a small time interval. Let us set
\begin{equation} \label{tildehy}
\widetilde{HY}^n_{kl}[D,N] = \frac{1}{ \left( \psi k_{n} \right)^{2}} \sum_{i = 0}^{n_{k} - k_{n} + 1} \sum_{j = 0}^{n_{l} - k_{n} + 1}
\widetilde \xi_{ij}^n
1_{ \{ (t_{i}^{k}, t_{i + k_n}^{k}] \cap (t_j^{l}, t_{j+k_n}^{l}] \neq \emptyset \}}.
\end{equation}
We first show that $\widetilde{HY}^n_{kl}[D,N]=o_{\mathbb P} (n^{-1/4})$. Observe that
\begin{equation*}
\E[|\widetilde{HY}^n_{kl}[D,N]|^2] =
\frac{1}{ \left( \psi k_{n} \right)^{4}} \sum_{i,i' = 0}^{n_{k} - k_{n} + 1} \sum_{j,j' = 0}^{n_{l} - k_{n} + 1}
\E \widetilde \xi_{ij}^n \widetilde \xi_{i'j'}^n
1_{ \{ (t_{i}^{k}, t_{i + k_n}^{k}] \cap (t_j^{l}, t_{j+k_n}^{l}] \neq \emptyset,
 (t_{i'}^{k}, t_{i' + k_n}^{k}] \cap (t_{j'}^{l}, t_{j'+k_n}^{l}] \neq \emptyset\}}.
\end{equation*}
Due to Lemma \ref{prooflem1} the above sum contains $O(nk_n^3)$ non-zero summands, because the $ \widetilde \xi_{ij}^n$'s are martingale
differences. Moreover, we have $\E[|\widetilde \xi_{ij}^n|^2]\leq C n^{-3/2}$ due to Lemma \ref{prooflem2}. Thus, we conclude $\E[|\widetilde{HY}^n_{kl}[D,N]|^2] \leq C n^{-1},$
which implies that $\widetilde{HY}^n_{kl}[D,N]=o_{\mathbb P}(n^{-1/4})$. In a second step we show that $HY^n_{kl}[D,N] - \widetilde{HY}^n_{kl}[D,N] =o_{\mathbb P}(n^{-1/4}).$ For this purpose, for any c\`agl\`ad bounded multivariate process $f$, we denote by $N_\delta^f(t)$ the number of jumps of $f$
bigger than $\delta >0$ before time $t$. Furthermore, we define
\bean
m_{\eta, \delta } (f)= \sup\{\|f_s - f_t\|:~ t\leq s\leq (t+\eta) \wedge 1, ~ N_\delta^f(s) - N_\delta^f(t)=0\}.
\eean
Roughly speaking, $m_{\eta, \delta } (f)$ is a modulus of continuity of $f$ on intervals of at most length $\eta$,
which do not contain jumps bigger than $\delta$. For $f$ as above, we obviously have
$\lim_{\delta \rightarrow 0} \limsup_{\eta\rightarrow 0} m_{\eta, \delta } (f)=0, \mathbb P-a.s.$
Observe that
\begin{equation*}
HY^n_{kl}[D,N] - \widetilde{HY}^n_{kl}[D,N] =
\frac{1}{ \left( \psi k_{n} \right)^{2}} \sum_{i = 0}^{n_{k} - k_{n} + 1} \sum_{j = 0}^{n_{l} - k_{n} + 1}
(\xi_{ij}^n - \widetilde \xi_{ij}^n)
1_{ \{ (t_{i}^{k}, t_{i + k_n}^{k}] \cap (t_j^{l}, t_{j+k_n}^{l}] \neq \emptyset \}}.
\end{equation*}
As we mentioned the above sum contains $O(nk_n)$ summands. We have
\begin{eqnarray*}
\Big|\overline D_{t_i^k}^k -  a_{t_i^k \wedge t_j^l}
\overline {\mbox{id}}_{t_i^k}^k\Big| \leq \sum_{h=1}^{k_n-1} \Big|g \Big(\frac{h}{k_n} \Big)\Big|
\int_{t_{i+h-1}^k}^{t_{i+h}^k} \|a_s - a_{t_i^k \wedge t_j^l} \| ds.
\end{eqnarray*}
The right-hand side of the above inequality is bounded since the process $a$ is bounded by $C n^{-1/2}$. Consequently, distinguishing between
the two situations, where $a$ has or does not have jumps bigger than $\delta$ in the interval $[t_{i+h-1}^k, t_{i+h}^k]$,
we obtain the inequality
\begin{eqnarray*}
\sum_{h=1}^{k_n-1} \Big|g \Big(\frac{h}{k_n} \Big)\Big|
\int_{t_{i+h-1}^k}^{t_{i+h}^k} \|a_s - a_{t_i^k \wedge t_j^l} \| ds
\leq
C n^{-1/2} \Big( m_{Ck_n/n, \delta } (a) + (\{N_\delta^a(t_{i+k_n}^k) - N_\delta^a(t_i^k \wedge t_j^l)\}
\wedge 1) \Big).
\end{eqnarray*}
Using  Lemma
\ref{prooflem2} and Cauchy-Schwarz inequality we deduce that
\begin{equation*}
n^{1/4}\E[|HY^n_{kl}[D,N] - \widetilde{HY}^n_{kl}[D,N]|] \leq
C \E \Big[ m_{Ck_n/n, \delta }^2 (a) + \Big(\frac{N_\delta^a(1)}{n} \wedge 1\Big)^2 \Big]^{1/2}.
\end{equation*}
Due to the dominated convergence theorem we conclude that
\begin{equation*}
\lim_{\delta \rightarrow 0} \limsup_{n\rightarrow \infty }
 \E \Big[ m_{Ck_n/n, \delta }^2 (a) + \Big(\frac{N_\delta^a(1)}{n} \wedge 1\Big)^2 \Big]=0.
\end{equation*}
Thus $HY^n_{kl}[D,N] - \widetilde{HY}^n_{kl}[D,N] =o_{\mathbb P}(n^{-1/4}).$ Summarizing all results of this section we get
\begin{equation*}
n^{1/4}(HY^n_{kl}[X] - HY^n_{kl}[N]) =o_{\mathbb P}(1)
\end{equation*}
meaning that the martingale part $N$ is the dominating term in the decomposition of $HY^n_{kl}[X]$. \qed

\subsubsection{A decomposition for the martingale part} \label{proofth2part3}

Having proved in the previous section that $HY^n[X]$ can be replaced by $HY^n[N]$ without affecting
the limit, we proceed with a further decomposition of $HY^n[N]$. In this section we will show that
$HY^n[N]$ is essentially an unbiased estimator of $\int_0^1 \Sigma_s ds$. Recall that
\begin{equation*}
HY^n_{kl}[N] = \frac{1}{ \left( \psi k_{n} \right)^{2}} \sum_{i = 0}^{n_{k} - k_{n} + 1} \sum_{j = 0}^{n_{l} - k_{n} + 1}
\overline N_{t_i^k}^k\overline N_{t_j^l}^l 1_{ \{ (t_{i}^{k}, t_{i + k_n}^{k}] \cap (t_j^{l}, t_{j+k_n}^{l}] \neq \emptyset \}}
\end{equation*}
By definition we have
\begin{eqnarray*}
\overline N_{t_i^k}^k\overline N_{t_j^l}^l = \sum_{h,h'=1}^{k_n-1} g \Big(\frac{h}{k_n} \Big)
g \Big(\frac{h'}{k_n} \Big) \De_{t_{i+h}^k} N^k \De_{t_{j+h'}^l} N^l
 \Big(1_{E_{ij}^{hh'}}
+1_{(E_{ij}^{hh'})^c} \Big)
\end{eqnarray*}
with
\begin{equation*}
 E_{ij}^{hh'}=\{(t_{i+h-1}^{k}, t_{i + h}^{k}] \cap (t_{j+h'-1}^{l}, t_{j+h'}^{l}] \neq \emptyset\}.
\end{equation*}
Now, we will write the above quantity as a sum of martingale differences plus bias. For this purpose
we need some additional notations. We decompose $ E_{ij}^{hh'}=
\cup_{r=1}^4 E_{ij}^{hh'}(r)$
with
\begin{eqnarray*}
E_{ij}^{hh'}(1) &=& \{(i,j),(h,h')|~ t_{j+h'-1}^l\geq t_{i+h-1}^k, \quad t_{j+h'}^l\geq t_{i+h}^k \}  \cap E_{ij}^{hh'} \\[1.5 ex]
E_{ij}^{hh'}(2) &=& \{(i,j),(h,h')|~ t_{j+h'-1}^l\geq t_{i+h-1}^k, \quad t_{j+h'}^l < t_{i+h}^k \}  \cap E_{ij}^{hh'} \\[1.5 ex]
E_{ij}^{hh'}(3) &=& \{(i,j),(h,h')|~ t_{j+h'-1}^l< t_{i+h-1}^k, \quad t_{j+h'}^l< t_{i+h}^k \}  \cap E_{ij}^{hh'} \\[1.5 ex]
E_{ij}^{hh'}(4) &=& \{(i,j),(h,h')|~ t_{j+h'-1}^l< t_{i+h-1}^k, \quad t_{j+h'}^l \geq  t_{i+h}^k \}  \cap E_{ij}^{hh'}
\end{eqnarray*}
On $E_{ij}^{hh'}(1)$  we deduce by It\^o formula:
\begin{eqnarray}
&&\De_{t_{i+h}^k} N^k \De_{t_{j+h'}^l} N^l = (N^k_{t_{j+h'-1}^l} -
N^k_{t_{i+h-1}^k}) \De_{t_{j+h'}^l} N^l + (N^k_{t_{i+h}^k}-N^k_{t_{j+h'-1}^l})(N^l_{t_{j+h'}^l} -
N^l_{t_{i+h}^k})  \nonumber \\[1.5 ex]
&&+ \int_{t_{j+h'-1}^l}^{t_{i+h}^k} (N^k_s - N^k_{t_{j+h'-1}^l}) dN_s^l +
\int_{t_{j+h'-1}^l}^{t_{i+h}^k} (N^l_s - N^l_{t_{j+h'-1}^l}) dN_s^k +
\int_{t_{j+h'-1}^l}^{t_{i+h}^k} \Sigma_s^{kl} ds \nonumber \\[1.5 ex]
&& = \sum_{r=1}^5 \mu_{ij}^{hh'}(1,r),
\end{eqnarray}
and similar decompositions are obtained on $E_{ij}^{hh'}(q)$, $q=2,3,4$, and we denote them by $\sum_{r=1}^5 \mu_{ij}^{hh'}(q,r)$.
Notice that all terms  $\mu_{ij}^{hh'}(q,r)$ are martingale differences for $1\leq q,r\leq 4$, while
$\mu_{ij}^{hh'}(q,5)$ gives the bias for all $1\leq q\leq 4$. We define
\begin{equation}
\mu_{ij}(q,r) = \sum_{h,h'=1}^{k_n-1} g \Big(\frac{h}{k_n} \Big)
g \Big(\frac{h'}{k_n} \Big) \mu_{ij}^{hh'}(q,r) 1_{E_{ij}^{hh'}(q)}
\end{equation}
for $1\leq q\leq 4,1\leq r\leq 5$. Now, a simple reordering shows that
\begin{eqnarray*}
&&\frac{1}{ \left( \psi k_{n} \right)^{2}} \sum_{i,j} \left( \sum_{q=1}^4 \mu_{ij}(q,5)  \right)
1_{ \{ (t_{i}^{k}, t_{i + k_n}^{k}] \cap (t_j^{l}, t_{j+k_n}^{l}] \neq \emptyset \}}=\frac{\sum_{h,h'=1}^{k_n-1} g \Big(\frac{h}{k_n} \Big)
g \Big(\frac{h'}{k_n} \Big)}{ \left( \psi k_{n} \right)^{2}} \int_0^1 \Sigma_s^{kl} ds + o_{\mathbb P}(n^{-1/4}) \\
&&= \int_0^1 \Sigma_s^{kl} ds + o_{\mathbb P}(n^{-1/4}),
\end{eqnarray*}
where the error in the first identity is due to border effects, and the second identity uses $\psi=\int_0^1 g(x) dx$.

Thus, we conclude that
\begin{equation} \label{etaest}
n^{1/4}\Big(HY^n_{kl}[N] - \int_0^1 \Sigma_s^{kl} ds \Big)=
\frac{n^{1/4}}{ \left( \psi k_{n} \right)^{2}} \sum_{i = 0}^{n_{k} - k_{n} + 1} \sum_{j = 0}^{n_{l} - k_{n} + 1}
\eta_{ij}^{kl}
1_{ \{ (t_{i}^{k}, t_{i + k_n}^{k}] \cap (t_j^{l}, t_{j+k_n}^{l}] \neq \emptyset \}} +o_{\mathbb P} (1),
\end{equation}
where
\begin{equation} \label{eta}
\eta_{ij}^{kl} = \overline{\mu}_{ij} + \sum_{q,r=1}^4
\mu_{ij}(q,r) ,
\end{equation}
\begin{equation} \label{overmu}
\overline{\mu}_{ij} = \sum_{h,h'=1}^{k_n-1} g \Big(\frac{h}{k_n} \Big)
g \Big(\frac{h'}{k_n} \Big) \De_{t_{i+h}^k} N^k \De_{t_{j+h'}^l} N^l 1_{(E_{ij}^{hh'})^c}.
\end{equation}
We remark again  all terms $\eta_{ij}^{kl}$ are now sums of martingale differences. \qed

\subsubsection{The approximation of the diffusion part II} \label{proofth2part7}
In this section we will justify the approximation
\begin{equation*}
n^{1/4}\Big(HY^n_{kl}[N] - \int_0^1 \Sigma_s^{kl} ds \Big) = M_n^{kl} (X,p) + \widetilde M_n^{kl} (X,p) + R_n^{kl}(p),
\end{equation*}
where $M_n (X,p)$ and $\widetilde M_n(X,p)$ are defined by (\ref{mndec}) and (\ref{nndec}) respectively,
for some $R_n^{kl}(p)$ with
\begin{equation} \label{Rn}
\lim_{p\rightarrow \infty } \limsup_{n\rightarrow \infty } ~\mathbb
P \Big(|R_n^{kl}(p)|>\delta \Big) =0
\end{equation}
for all $\delta >0$. This means that the diffusion part $n^{1/4}\Big(HY^n_{kl}[N] - \int_0^1 \Sigma_s^{kl} ds \Big) $
of our statistic is asymptotically equivalent to the sum of the diffusion parts of big and small blocks. Recalling
the estimate (\ref{etaest}) from the previous section, it is easy to show
\begin{eqnarray*}
&&R_n^{kl}(p)= n^{1/4}\Big(HY^n_{kl}[N] - \int_0^1 \Sigma_s^{kl} ds \Big) - M_n^{kl} (X,p) - \widetilde M_n^{kl} (X,p) \\[1.5 ex]
&&= \frac{n^{1/4}}{ \left( \psi k_{n} \right)^{2}} \sum_{i,j}
(\eta_{ij}^{kl} - \widetilde{\eta}_{ij}^{kl})
1_{ \{ (t_{i}^{k}, t_{i + k_n}^{k}] \cap (t_j^{l}, t_{j+k_n}^{l}] \neq \emptyset \}} +o_{\mathbb P} (1),
\end{eqnarray*}
where $\widetilde{\eta}_{ij}^{kl}$ is defined in the same way as $\eta_{ij}^{kl}$ (see (\ref{eta})) except
the process $N^k$ (resp. $N^l$) is replaced by $(\sigma_{\min A_z(p)} W)^k$ (resp. $(\sigma_{\min A_{z'}(p)} W)^l$)
when $t_i^k\in A_z(p)$ for some $z$ (resp. $t_j^k\in A_{z'}(p)$ for some $z'$) and $A=B$ or $A=S$.
Note that the only difference compared to proving (\ref{etaest}) lies in the fact that $M_n^{kl} (X,p) + \widetilde M_n^{kl} (X,p)$ is unbiased by construction.

Recall that the quantity $\eta_{ij}^{kl}$ (resp. $\widetilde{\eta}_{ij}^{kl}$) consists of 17 summands. Hence, we have the
decomposition
$R_n^{kl}(p) = \sum_{r=1}^{17} R_n^{kl}(p,r).$
As an example
we will only consider the treatment of the first summand, i.e.
\begin{equation*}
R_n^{kl}(p,1) = \frac{n^{1/4}}{ \left( \psi k_{n} \right)^{2}} \sum_{i,j}
(\overline{\mu}_{ij} - \widetilde{\overline{\mu}}_{ij})
1_{ \{ (t_{i}^{k}, t_{i + k_n}^{k}] \cap (t_j^{l}, t_{j+k_n}^{l}] \neq \emptyset \}},
\end{equation*}
where $\overline{\mu}_{ij}$ is defined by (\ref{overmu}). We conclude that
\begin{eqnarray*}
&&\E[|\overline{\mu}_{ij} - \widetilde{\overline{\mu}}_{ij}|^2]=
\E\Big[ \sum_{h,h',q,q'} g \Big(\frac{h}{k_n} \Big)
g \Big(\frac{h'}{k_n} \Big) g \Big(\frac{q}{k_n} \Big)
g \Big(\frac{q'}{k_n} \Big) \De_{t_{i+h}^k} (N-\sigma_{\min A_z(p)} W)^k \\[1.5 ex]
&&\times  \De_{t_{j+h'}^l} (N-\sigma_{\min A_{z'}(p)} W)^l
\De_{t_{i+q}^k} (N-\sigma_{\min A_z(p)} W)^k \\[1.5 ex]
&& \times  \De_{t_{j+q'}^l} (N-\sigma_{\min A_{z'}(p)} W)^l  1_{(E_{ij}^{hh'})^c} 1_{(E_{ij}^{qq'})^c} \Big],
\end{eqnarray*}
where $1\leq h,h',q,q'\leq k_n$ and either $h=q, h'=q'$ or
$$(t_{i+h-1}^k, t_{i+h}^k]\cap (t_{j+q'-1}^l, t_{j+q'}^l] \not= \emptyset, \qquad
(t_{i+q-1}^k, t_{i+q}^k]\cap (t_{j+h'-1}^l, t_{j+h'}^l] \not= \emptyset,$$
as otherwise the expectation vanishes. We remark that the above sum contains $O(k_n^2)$ terms.
Now we follow the same strategy as in Section \ref{proofth2part4}. First, we note that
\begin{equation*}
\E[|R_n^{kl}(p,1)|^2] = \frac{n^{1/2}}{ \left( \psi k_{n} \right)^{4}} \sum_{i,j,i',j'} \E
(\overline{\mu}_{ij} - \widetilde{\overline{\mu}}_{ij}) (\overline{\mu}_{i'j'} - \widetilde{\overline{\mu}}_{i'j'})
1_{ \{ (t_{i}^{k}, t_{i + k_n}^{k}] \cap (t_j^{l}, t_{j+k_n}^{l}] \neq \emptyset,
(t_{i'}^{k}, t_{i' + k_n}^{k}] \cap (t_{j'}^{l}, t_{j'+k_n}^{l}] \neq \emptyset\}},
\end{equation*}
where the number of non-zero summands is $O(nk_n^3)$. Using the Cauchy-Schwarz inequality and the same approximations
as at the end of Section \ref{proofth2part4}, we deduce that
\begin{equation*}
\E[|R_n^{kl}(p,1)|^2] \leq C \E \Big[ m_{pk_n/n, \delta }^2 (\sigma) + \Big(\frac{N_\delta^\sigma (1)}{n} \wedge 1\Big)^2 \Big]
\end{equation*}
for any $\delta >0$. Thus, for any fixed $p$, we have (by choosing $n$ large and then $\delta$ small) $
\lim_{n\rightarrow \infty}
\E[|R_n^{kl}(p,1)|^2] =0.$
Hence, (\ref{Rn}) for any $\delta >0$, and we are done. \qed

\subsubsection{The approximation of the mixed part} \label{proofth2part5}
In this section we will prove that
\begin{equation*}
n^{1/4}HY^n_{kl}[X,\varepsilon ]= M_n^{kl} (X,\varepsilon ,p) + \widetilde M_n^{kl} (X,\varepsilon ,p) + R_n^{kl}(p),
\end{equation*}
where $M_n (X,\varepsilon ,p)$ and $\widetilde M_n(X,\varepsilon ,p)$ are defined by (\ref{mndec}) and (\ref{nndec}) respectively,
$HY^n_{kl}[X,\varepsilon ]$ is given by (\ref{hydec}) and some  $R_n^{kl}(p)$ with
(\ref{Rn}) for all $\delta >0$. This proof is easier than the proofs in previous sections, because the processes $X$ and $\varepsilon$
are independent. We first show that
\begin{equation*}
n^{1/4} HY^n_{kl}[D,\varepsilon ]  =
\frac{n^{1/4}}{ \left( \psi k_{n} \right)^{2}} \sum_{i = 0}^{n_{k} - k_{n} + 1} \sum_{j = 0}^{n_{l} - k_{n} + 1}
\Big(
\overline D_{t_i^k}^k\overline {\varepsilon }_{t_j^l}^l + \overline {\varepsilon}_{t_i^k}^k\overline D_{t_j^l}^l
\Big) 1_{ \{ (t_{i}^{k}, t_{i + k_n}^{k}] \cap (t_j^{l}, t_{j+k_n}^{l}] \neq \emptyset \}}
\end{equation*}
is a negligible sequence. Using Lemma \ref{prooflem2} and proceeding as in the treatment of the term $\widetilde{HY}^n_{kl}[D,N]$
from (\ref{tildehy}) we deduce that
$\E[|HY^n_{kl}[D,\varepsilon ]|^2]\leq Cn^{-1}.$
Hence,
$n^{1/4} HY^n_{kl}[D,\varepsilon ]  \toop 0.$
Next, we obtain that
\begin{eqnarray*}
&&R_n^{kl}(p) =n^{1/4}HY^n_{kl}[N,\varepsilon ]- M_n^{kl} (X,\varepsilon ,p) - \widetilde M_n^{kl} (X,\varepsilon ,p) + o_{\mathbb P} (1)
\\[1.5 ex]
&&= \frac{n^{1/4}}{ \left( \psi k_{n} \right)^{2}} \sum_{i,j}
\Big(
\overline {(N-\sigma_{\min A_z(p)} W)}_{t_i^k}^k\overline {\varepsilon }_{t_j^l}^l +
\overline {\varepsilon}_{t_i^k}^k\overline {(N-\sigma_{\min A_{z'}(p)} W)}_{t_j^l}^l
\Big) 1_{ \{ (t_{i}^{k}, t_{i + k_n}^{k}] \cap (t_j^{l}, t_{j+k_n}^{l}] \neq \emptyset \}} + o_{\mathbb P} (1)
\end{eqnarray*}
Using again Lemma \ref{prooflem2}, the independence between $\varepsilon$ and the components of $X$, and similar methods as
for $R_n^{kl}(p,1)$ in the previous section, we conclude that
\begin{equation*}
\E[|R_n^{kl}(p)|^2] \leq C \E \Big[ m_{pk_n/n, \delta }^2 (\sigma) + \Big(\frac{N_\delta^\sigma (1)}{n} \wedge 1\Big)^2 \Big]
\end{equation*}
for any $\delta >0$. Thus, for any fixed $p$, we have $\lim_{n\rightarrow \infty}
\E[|R_n^{kl}(p,1)|^2] =0,$
and hence (\ref{Rn})
for any $\delta >0$, and we are done. \qed

\subsubsection{The noise part and the final identity} \label{proofth2part6}
Finally, we will show that
\begin{equation*}
n^{1/4}HY^n_{kl}[\varepsilon ]  = M_n^{kl} (\varepsilon,p) +
\widetilde M_n^{kl} (\varepsilon,p) + R_n^{kl}(p),
\end{equation*}
where $M_n (\varepsilon,p)$ and $\widetilde M_n(\varepsilon,p)$ are defined by (\ref{mndec}) and (\ref{nndec}) respectively,
for some $R_n^{kl}(p)$ with (\ref{Rn}) for all $\delta >0$. This is a relatively easy exercise, because by definition we just need to prove that
$n^{1/4} \E[HY^n_{kl}[\varepsilon ]] = o(1).$
By reordering the statistic $HY^n_{kl}$ we obtain that
\begin{equation*}
n^{1/4} \E[HY^n_{kl}[\varepsilon ]] = \frac{n^{1/4}}{ \left( \psi k_{n} \right)^{2}} \E \Big[ \sum_{i,j: ~t_i^k = t_j^l}
a_{ij}^{kl}(n)\varepsilon_{t_i^k}^k
\varepsilon_{t_j^l}^l \Big]
\end{equation*}
for some constants $a_{ij}^{kl}(n)$ with $|a_{ij}^{kl}(n)|\leq C$. A simple calculation shows that
$$a_{ij}^{kl}(n) = \left( \sum_{j=0}^{k_n-1} g \Big(\frac{j+1}{k_n} \Big) -g \Big(\frac{j}{k_n} \Big) \right)^2 = (g(1)-g(0))^2=0$$
except for those $t_i^k$ and $t_j^l$ that are among the first and last $O(n^{1/2})$ summands. Hence, $n^{1/4} \E[HY^n_{kl}[\varepsilon ]] = o(1)$
and we deduce that
\begin{equation*}
\lim_{p\rightarrow \infty } \limsup_{n\rightarrow \infty } ~\mathbb
P \Big(|n^{1/4}HY^n_{kl}[\varepsilon ]  - M_n^{kl} (\varepsilon,p)-
\widetilde M_n^{kl} (\varepsilon,p)|>\delta \Big) =0
\end{equation*}
for all $\delta >0$. \\ \\
Finally, let us put things together. In Sections \ref{proofth2part4}--\ref{proofth2part6} we have proved the identity
\begin{equation*}
L^n = n^{1/4}(HY^n - [X]) = M_n (p) + \widetilde M_n(p) + R_n(p)
\end{equation*}
for some $R_n(p)$ and we have shown (see Section \ref{proofth2part2}) that
\begin{equation*}
\lim_{p\rightarrow \infty } \limsup_{n\rightarrow \infty } ~\mathbb
P \Big(|\widetilde M_n(p)| + |R_n(p)|>\delta \Big) =0
\end{equation*}
for all $\delta >0$. On the other hand, we have proved in Section \ref{proofth2part1} that
\begin{equation*}
M_n (p) \stab M(p)=MN(0, V_p) \qquad \mbox{as~~} n\rightarrow \infty
\end{equation*}
and, for $p\rightarrow \infty$:
\begin{equation*}
V_p \toop V, \qquad M(p) \toop L=MN(0, V).
\end{equation*}
This completes the proof of Theorem \ref{th2}. \qed

\subsection{Consistency of the variance estimators} \label{consva}

\subsubsection{Proof of Theorem \ref{th3}} \label{proofth3}
It is obviously enough to prove the result for the unsymmetrized estimator
\bean
\widetilde{V}^{n,1}_{kl,k'l'} = \sqrt n \sum_{\alpha=1}^{[ \frac{n}{\beta_n} ]-1} \Big(HY^n_{kl}(\al) HY^n_{k'l'}(\al) - HY^n_{kl}(\al) HY^n_{k'l'}(\al-1)\Big)
\eean
only, and we introduce two approximating versions of $HY^n_{kl}(\al)$ first, namely
\bean
&&\widetilde{HY}^n_{kl}(\al) = \frac{1}{ \left( \psi k_{n} \right)^{2}} \sum_{t_i^k \in B_n(\alpha)} \sum_{j = 0}^{n_{l} - k_{n} + 1}
\overline Z(\alpha)_{t_i^k}^k\overline Z(\alpha)_{t_j^l}^l 1_{ \{ (t_{i}^{k}, t_{i + k_n}^{k}] \cap (t_j^{l}, t_{j+k_n}^{l}] \neq \emptyset \}}, \\
&&\overline{HY}^n_{kl}(\al) = \frac{1}{ \left( \psi k_{n} \right)^{2}} \sum_{t_i^k \in B_n(\alpha)} \sum_{j = 0}^{n_{l} - k_{n} + 1}
\overline Z(\alpha-1)_{t_i^k}^k\overline Z(\alpha-1)_{t_j^l}^l 1_{ \{ (t_{i}^{k}, t_{i + k_n}^{k}] \cap (t_j^{l}, t_{j+k_n}^{l}] \neq \emptyset \}},
\eean
where we have set
\bean
\overline {Z(\alpha)}_{t_i^k}^k = \overline \varepsilon_{t_i^k}^k + \sum_{\nu=1}^d \sigma_{\frac{\alpha \beta_n}{n}}^{k \nu} \overline {W^{v}}_{t_i^k}^k
\eean
as in (\ref{basicappr}), and the $W^{\nu}$ denote the independent components of the $d'$-dimensional Brownian motion $W$. Since $\sigma$ is assumed to be an It\^o semimartingale itself, the error due to replacing $\overline {Y}_{t_i^k}^k$ by $\overline {Z(\alpha)}_{t_i^k}^k$ is small: Let $t_i^k \in B_n(\alpha)$. Then
\bean
E|\overline {Y}_{t_i^k}^k-\overline {Z(\alpha)}_{t_i^k}^k| &=& E \Big| \sum_{j=1}^{k_n-1} g (j/k_n) \Big( \De_{t_{i+j}^k} D^k + \sum_{\nu=1}^d \int_{\frac{i+j-1}{n}}^{\frac{i+j}{n}} (\sigma_s^{k\nu} - \sigma_{\frac{\alpha \beta_n}{n}}^{k\nu}) dW^{\nu}_s\Big) \Big| \\ &\leq& C \Big(\frac{k_n}{n} + \Big( \sum_{j=1}^{k_n-1} g^2(j/k_n) \sum_{\nu=1}^d \E \Big|\int_{\frac{i+j-1}{n}}^{\frac{i+j}{n}} (\sigma_s^{k\nu} - \sigma_{\frac{\alpha \beta_n}{n}}^{k\nu}) dW^{\nu}_s\Big|^2 \Big)^{1/2} \Big) \\ &\leq& C \Big( \frac{k_n}n + \Big( k_n \frac 1n \frac{\be_n}n \Big)^{1/2} \Big) \leq C \frac{\sqrt {k_n \beta_n}}{n}.
\eean
Lemma \ref{prooflem1} and Lemma \ref{prooflem2} give $E|HY^n_{kl}(\al)| \leq C \beta_n/n$, thus it is simple to deduce $E|HY^n_{kl}(\al) - \widetilde{HY}^n_{kl}(\al)| \leq C (\beta_n/n)^{3/2}$, and analogously for $\overline{HY}^n_{kl}(\al)$, so using $\eta < 2/3$ we obtain $\widetilde{V}^{n,1}_{kl,k'l'} - \overline{V}^{n,1}_{kl,k'l'} = o_{\mathbb P} (1)$ with
\bean
\overline{V}^{n,1}_{kl,k'l'} = \sqrt n \sum_{\alpha=1}^{[ \frac{n}{\beta_n} ]} \Big(\overline{HY}^n_{kl}(\al) \overline{HY}^n_{k'l'}(\al) - \overline{HY}^n_{kl}(\al) \widetilde {HY}^n_{k'l'}(\al-1)\Big).
\eean
The remainder of the proof is simple now. Without loss of generality let $\beta_n > 4bk_n$ hold, so only $\overline{HY}^n_{k'l'}(\al)$ and $\overline{HY}^n_{k'l'}(\al+1)$ might share increments of $Y$. Then we obtain
\bean
&& \sqrt n \Big| \sum_{\alpha=1}^{[ \frac{n}{\beta_n} ]} \E \Big( \overline{HY}^n_{kl}(\al) \overline{HY}^n_{k'l'}(\al) - \E[\overline{HY}^n_{kl}(\al) \overline{HY}^n_{k'l'}(\al)|\mathcal F_{\frac{(\alpha-1) \beta_n}n}]\Big) \Big| \leq C \frac{\beta_n^{3/2}}{n}, \\
&& \sqrt n \Big| \sum_{\alpha=1}^{[ \frac{n}{\beta_n} ]} \E \Big( \overline{HY}^n_{kl}(\al) \widetilde {HY}^n_{k'l'}(\al-1) - \E[\overline{HY}^n_{kl}(\al) \widetilde {HY}^n_{k'l'}(\al-1)|\mathcal F_{{\frac{(\alpha-1) \beta_n}n}}]\Big) \Big| \leq C \frac{\beta_n^{3/2}}{n},
\eean
by conditional independence, and we are left with
\bean
\overline{V}^{n,1}_{kl,k'l'} = \sqrt n \sum_{\alpha=1}^{[ \frac{n}{\beta_n} ]} E[\overline{HY}^n_{kl}(\al) \overline{HY}^n_{k'l'}(\al) - \overline{HY}^n_{kl}(\al) \widetilde {HY}^n_{k'l'}(\al-1)|\mathcal F_{\frac{(\alpha-1) \beta_n}n}] + o_{\mathbb P} (1).
\eean
Write $V_{kl,k'l'} = \int_0^1 r_u du$, where the process $r$ is given by the right hand side of (\ref{variance}). From the same arguments as in Lemma \ref{form} and Lemma \ref{form2} in the Appendix plus using $\eta > 1/2$ we obtain
\bean
\sqrt n E[\overline{HY}^n_{kl}(\al) \overline{HY}^n_{k'l'}(\al) - \overline{HY}^n_{kl}(\al) \widetilde {HY}^n_{k'l'}(\al-1)|\mathcal F_{\frac{(\alpha-1) \beta_n}n}] = \int_{\frac{\alpha \beta_n}{n}}^{\frac{(\alpha+1) \beta_n}{n}} r(u) du + o(\frac{\be_n}n),
\eean
uniformly in $\alpha$, and the proof is complete. \qed

\subsubsection{Proof of Theorem \ref{th4}} \label{proofth4}
From the proof of Theorem \ref{th1} we have
\bean
HY^n([0,s]) - HY^n([0,s-l_n]) - \int_{s-l_n}^{s} \Sigma_u du = o_{\mathbb P} (l_n),
\eean
uniformly in $s$. Therefore the discussion on $\Psi_{n}^{kl}$ shows that we are left to prove
\bean
\int_{l_n}^{1} \Big( \frac{\int_{s-l_n}^{s} \Sigma_u du}{l_n} - \Sigma_s \Big) ds = o_{\mathbb P} (1),
\eean
which by left-continuity is obvious as well. \qed

\subsubsection{Proof of Theorem \ref{th5}} \label{proofth5}
All we need to prove is
\bean
\frac{\kappa}{3 \theta \mu^2} \sum_{i=1}^{n-k_n+1} |\overline Y_{t_i}|^4 \toop \kappa \theta \int_0^1 \frac{\sigma_u^4}{f'(u)} du  +
\frac{2 \kappa \tilde \mu}{\theta \mu} \Psi \int_0^1 \sigma_u^2 du + \frac{\kappa \tilde \mu^2}{\theta^{3} \mu^2} \Psi^2.
\eean
Since $\sigma$ is c\`agl\`ad, we know from the proof of Theorem 1 in \cite{PV1} that we may replace $|\overline Y_{t_i}|^4$ by $|\sigma_{t_i} \overline W_{t_i} + \overline \varepsilon_{t_i}|^4 $ without affecting the limit. We have
\bean
\frac{2\kappa}{3 \theta  \mu^2} \sum_{i=1}^{n-k_n+1} \sigma_{t_i}^4 \E[|\overline W_{t_i}|^4] = \frac{2\kappa}{\theta} \frac{k_n^2}{n^2}\sum_{i=1}^{n-k_n+1} \sigma_{t_i}^4 + o_{\mathbb P} (1) = 2\kappa \theta \frac{1}{n}\sum_{i=1}^{n-k_n+1} \sigma_{t_i}^4 + o_{\mathbb P} (1),
\eean
and similar identities hold for $6 |\overline W_{t_i}|^2 |\overline \varepsilon_{t_i}|^2$ and $|\overline \varepsilon_{t_i}|^4$ as well. The result follows easily now from a Riemann approximation. \qed

\section{Appendix} \label{append}
\setcounter{equation}{0}
\renewcommand{\theequation}{\thesection.\arabic{equation}}

In this final paragraph we discuss the computation of the asymptotic (conditional) variance $V_p$ from Theorem \ref{proofth1}, which amounts to showing step (i) of its proof, and to prove convergence of $V_p$ to the final variance $V$ afterwards. We start with a decomposition of $\zeta_{zn}^{kl} (p)$ into a pure diffusion part, two mixed parts and a noise one, as we write
\bean
~ \zeta_{zn}^{kl} (p) = \sum_{s=1}^3 \zeta_{zn}^{kl} (s,p), \qquad \zeta_{zn}^{kl} (s,p) = n^{1/4} \sum_{t_i^k, t_j^l\in B_z(p)} \Big(\alpha_{ij}^{kl} (s,p) - \E[\alpha_{ij}^{kl} (s,p)| \mathcal F_{\min B_z(p)}] \Big),
\eean
with
\bean
&& \alpha_{ij}^{kl} (1,p) =  \frac{1}{ \left( \psi k_{n} \right)^{2}}
(\overline{\sigma_{\min B_z(p)} W})  _{t_i^k}^k
(\overline{\sigma_{\min B_z(p)} W})  _{t_j^l}^l
1_{ \{ (t_{i}^{k}, t_{i + k_n}^{k}] \cap (t_j^{l}, t_{j+k_n}^{l}] \neq \emptyset \}}, \\
&& \alpha_{ij}^{kl} (2,p) =  \frac{1}{ \left( \psi k_{n} \right)^{2}}
[(\overline{\sigma_{\min B_z(p)} W})  _{t_i^k}^k
\overline \varepsilon  _{t_j^l}^l + \overline \varepsilon  _{t_i^k}^k
(\overline{\sigma_{\min B_z(p)} W})  _{t_j^l}^l]
1_{ \{ (t_{i}^{k}, t_{i + k_n}^{k}] \cap (t_j^{l}, t_{j+k_n}^{l}] \neq \emptyset \}}, \\
&& \alpha_{ij}^{kl} (3,p) =  \frac{1}{ \left( \psi k_{n} \right)^{2}}
\overline \varepsilon  _{t_i^k}^k
\overline \varepsilon  _{t_j^l}^l
1_{ \{ (t_{i}^{k}, t_{i + k_n}^{k}] \cap (t_j^{l}, t_{j+k_n}^{l}] \neq \emptyset \}}. \\
\eean
By independence of $W$ and $\varepsilon$ it suffices to discuss $$V_{p}^{kl,k'l'}(s) = \sum_{z} \E[\zeta_{zn}^{kl} (s,p) \zeta_{zn}^{k'l'} (s,p)| \mathcal F_{\min B_z(p)}]$$ with $s = 1,2,3$ only, and the final variance $V_p^{kl,k'l'}$ is the sum of the three limits in probability. Throughout each of the next subsections we also write
\bean
\beta_{ijqr}^{klk'l'} (s,p) = \Big(\alpha_{ij}^{kl} (s,p) - \E[\alpha_{ij}^{kl} (s,p)| \mathcal F_{\min B_z(p)}] \Big) \Big(\alpha_{qr}^{k'l'} (s,p) - \E[\alpha_{qr}^{k'l'} (s,p)| \mathcal F_{\min B_z(p)}] \Big),
\eean
and we introduce the auxiliary interval
\bean
\tilde B_z(p)= \Big[\frac{z(p+b)k_n + 2bk_n}{n}, \frac{z(p+b)k_n +(p-2b)k_n}{n}\Big),
\eean
which is slightly smaller than $B_z(p)$, but their sizes become close as $p$ grows eventually. Without loss of generality let $p$ be large enough for $\tilde B_z(p)$ to be non-empty.

\subsection{The contribution of the diffusion to the variance} \label{proofth1part1}

We begin with the pure diffusion part of the variance. By definition, we have
\bea \label{alpharep}
\alpha_{ij}^{kl} (1,p) = \frac{1}{ \left( \psi k_{n} \right)^{2}} \sum_{\nu_1, \nu_2 = 1}^{d'} \sigma^{k\nu_1}_{\min B_z(p)} \sigma^{l\nu_2}_{\min B_z(p)} \overline {W^{\nu_1}}_{t_i^k} \overline {W^{\nu_2}}_{t_j^l}
1_{ \{ (t_{i}^{k}, t_{i + k_n}^{k}] \cap (t_j^{l}, t_{j+k_n}^{l}] \neq \emptyset \}}.
\eea
In the following we will simply write $\sigma$ instead of $\sigma_{\min B_z(p)}$, whenever the particular time is obvious. Recall (\ref{preave}). Setting
\bean
&& F_{z,p}(k,l) = \{(i,j): \exists u,v \in \{1, \ldots, k_n\} \mbox{ with } t_{i-u}^{k} \in B_z(p), t_{j-v}^{k} \in B_z(p) \}, \\
&& \tilde F_{z,p}(k,l) = \{(i,j) \in F_{z,p}(k,l): t_{i}^{k} \in \tilde B_z(p) \},
\eean
we write
\bea \label{cdef}
\sum_{t_i^k, t_j^l\in B_z(p)} \overline {W^{\nu_1}}_{t_i^k} \overline {W^{\nu_2}}_{t_j^l}
1_{ \{ (t_{i}^{k}, t_{i + k_n}^{k}] \cap (t_j^{l}, t_{j+k_n}^{l}] \neq \emptyset \}} = \sum_{(i,j) \in F_{z,p}(k,l)} c^n_{ij}(k,l) \De_{t_{i}^k} W^{\nu_1} \De_{t_{j}^l} W^{\nu_2} \quad
\eea
for certain numbers $c^n_{ij}(k,l)$ depending on the function $g$. These constants count how often and with which weight a particular product $\De_{t_{i}^k} W^{\nu_1} \De_{t_{j}^l} W^{\nu_2}$ appears in $\alpha_{ij}^{kl} (1,p)$. Let us start with a simple lemma.

\begin{lem} \label{covariance1}
We have
\bea \label{cova} \nonumber
&&\E[\zeta_{zn}^{kl} (1,p) \zeta_{zn}^{k'l'} (1,p)| \mathcal F_{\min B_z(p)}] = \frac{n^{1/2}}{(\psi k_n)^4} \sum_{(i,j) \in F_{z,p}(k,l)} \sum_{(q,r) \in F_{z,p}(k',l')} c^n_{ij}(k,l) c^n_{qr}(k',l') \\ &&\sum_{v_1, v_2=1}^d \Big(\si^{kv_1} \si^{lv_2} \si^{k'v_1} \si^{l'v_2} \E[\De_i^{n_k} W^{v_1} \De_q^{n_{k'}} W^{v_1}] \E[\De_j^{n_{l}} W^{v_2} \De_r^{n_{l'}} W^{v_2}]  \\ &&\hspace{2cm}+\si^{kv_1} \si^{lv_2} \si^{k'v_2} \si^{l'v_1} \E[\De_i^{n_k} W^{v_1} \De_r^{n_{l'}} W^{v_1}] \E[\De_j^{n_l} W^{v_2} \De_q^{n_{k'}} W^{v_2}]\Big). \nonumber
\eea
\end{lem}
{\it Proof:} We have to compute
\bean
n^{1/2} \sum_{t_i^k, t_j^l\in B_z(p)} \sum_{t_q^{k'}, t_r^{l'}\in B_z(p)} \E[\beta_{ijqr}^{klk'l'} (1,p) | \mathcal F_{\min B_z(p)}],
\eean
and we begin with the conditional expectation of $\alpha_{ij}^{kl} (1,p) \alpha_{qr}^{k'l'} (1,p)$. Using the representations in (\ref{alpharep}) and (\ref{cdef}) plus measurability of $\sigma$ all we have to compute is
$
\E[\De_{t_{i}^k} W^{\nu_1} \De_{t_{j}^l} W^{\nu_2} \De_{t_{q}^{k'}} W^{\nu_3} \De_{t_{r}^{l'}} W^{\nu_4}].
$
Apply the well-known property
$\E[N_1N_2N_3N_4] = \E[N_1N_2]\E[N_3N_4] + \E[N_1N_3] \E[N_2N_4] + \E[N_1N_4] \E[N_2N_3]$
for a (centred) normal variable $(N_1,N_2,N_3,N_4)$. As $W^{\nu_1}$ and $W^{\nu_2}$ are independent for $\nu_1 \neq \nu_2$, the conditional expectation of $\alpha_{ij}^{kl} (1,p) \alpha_{qr}^{k'l'} (1,p)$ becomes the right hand side of (\ref{cova}) plus a third term, which is easily identified as the product of $\E[\alpha_{ij}^{kl} (1,p)| \mathcal F_{\min B_z(p)}]$ and $\E[\alpha_{qr}^{k'l'} (1,p)| \mathcal F_{\min B_z(p)}]$. This gives the result. \qed \\ \\
Using the previous lemma, the main part of the remainder consists in a computation of the constants $c^n_{ij}(k,l)$. Let us keep $i$ with $t_i^k \in B_z(p)$ fixed for the moment and define various auxiliary quantities, namely
\bean
\tilde j = [n_l f_l(t_{i-k_n}^{k})], \quad j' = [n_l f_l(t_i^{k})], \quad \bar j = [n_l f_l(t_{i+k_n}^{k})].
\eean
These quantities obviously depend on $i$ and $n$, even though it does not appear in the notation, and their use is to relate observation times in the $l$th grid to those in the $k$th one. For example, $j'$ is the largest index $j$ such that $t_j^{l}$ is left of $t_i^{k}$, and $\tilde j$ and $\bar j$ play similar roles. There are two observations to be made: First, in order for $c^n_{ij}(k,l)$ to be non-zero, the condition
\bee \label{jterms}
\tilde j - k_n + 1 \leq j \leq \bar j + k_n - 1
\eee
has to hold. This is an easy consequence of the fact that $t_{i-k_n}^{k} < t_{j+k_n-1}^{l}$ and $t_{i+k_n-1}^{k} > t_{j-k_n}^{l}$ need to be satisfied in order for the product of the corresponding increments of $Y^k$ and $Y^l$ to appear in $HY^n$. Second, it is not obvious that $\tilde j - k_n + 1 $ and $\bar j + k_n -1$ correspond to time points of $B_z(p)$ as well. However, by definition of $b$ we know that they do if $t_i^k$ belongs to $\tilde B_z(p)$, as for example $t_{i-k_n}^{k}$ lies within $[t_i^k - \frac{bk_n}{n}, t_i^k)$ and thus $t_{\tilde j-k_n-1}^{l} \in [t_i^k - \frac{2bk_n}{n}, t_i^k)$. Let us focus on this case for a moment, as these terms are responsible for the main contribution to $V_p$.

\begin{lem} \label{cij} Assume that we have $t_i^k \in \tilde B_z(p)$ and recall the definition of the functions $h_{kl}$ and $\psi$ in (\ref{hfun}) and (\ref{psifun}). Then we have, uniformly for all $(i,j)$ that satisfy (\ref{jterms}),
\bee \label{c}
c^n_{ij}(k,l) = k_n^2 \psi\Big( \frac{n_l f_l(t_i^{k})-j}{k_n}, h_{kl}(t_i^{k}) \Big) + o(k_n^2).
\eee
\end{lem}
\textit{Proof:} One singles out four cases for $j$ and computes $c^n_{ij} = c^n_{ij}(k,l)$ for each of these separately. For example,
\bean
\tilde j -k_n +1 \leq j \leq \tilde j \quad \mbox{gives} \quad c^n_{ij} = \sum_{l_1 = 1}^{j-1-(\tilde j - k_n)} \sum_{l_2 = \max(i+1 - [n_k f_k (t^{l}_{j+k_n-l_1})],1)}^{k_n} g(l_1/k_n) g(l_2/k_n),
\eean
all identities up to a possible error of (uniform) order $k_n$. This can be seen as follows: First, the choice of $l_1$ is limited, as $g(l_1/k_n)$ comes from $\overline W_{t_{j-l_1}^l}$ which involves $\De_{t_j^l}^{n_l} W$ as its $l_1$th summand. If $l_1$ is small, then at least some pre-averaged statistics in the $k$th grid starting left of $t^k_{[n_k f_k (t^{l}_{j+k_n-l_1})]}$ intersect with $\overline W_{t_{j-l_1}^l}$ and include $\De_{t_i^k}^{n_k} W$, and those ones are responsible for $g(l_2/k_n)$. On the other hand, if $l_1$ is $j-(\tilde j - k_n)$ or larger, then the corresponding $\overline W_{t_{j-l_1}^l}$ has only empty intersections with any pre-averaged statistic in the $k$th grid involving $\De_i^{n_k} W$. Similar arguments hold in the other situations, as
\bean
 \tilde j < j \leq j' \quad \mbox{gives} \quad 	c^n_{ij} = \sum_{l_1 = 1}^{k_n} \sum_{l_2 = \max(i+1 - [n_k f_k (t^{l}_{j+k_n-l_1})],1)}^{k_n} g(l_1/k_n) g(l_2/k_n), &&\\
j' < j < \bar j \quad \mbox{gives} \quad 	c^n_{ij} = \sum_{l_1 = 1}^{k_n} \sum_{l_2 = 1}^{\min(k_n + i -1 - [n_k f_k (t^{l}_{j-l_1})],k_n)} g(l_1/k_n) g(l_2/k_n), &&\\
\quad \bar j \leq j \leq \bar j + k_n - 1\quad \mbox{gives} \quad 	c^n_{ij} = \sum_{l_1 = j - \bar j +1}^{k_n} \sum_{l_2 = 1}^{\min(k_n +i -1 - [n_k f_k (t^{l}_{j-l_1})],k_n)} g(l_1/k_n) g(l_2/k_n). &&\\
\eean
One can forget about minimum and maximum in the second sums, because $g$ vanishes outside of $[0,1]$ anyway. Have a look at the first expression now. For $l_1 \geq j-(\tilde j - k_n)$ we obtain by monotonicity
\bean
i+1 - [n_k f_k (t^{l}_{j+k_n-l_1})] &\geq& i +1 - [n_k f_k (t^{l}_{j+k_n-(j-(\tilde j - k_n))})] = i +1 - [n_k f_k (t^{l}_{\tilde j})] \\ &\geq& i +1- (i-k_n) = k_n+1.
\eean
By assumption on $g$ again we see that the sum over $l_1$ in the first expression for $c^n_{ij}$ may thus be allowed to run to $k_n$ as well, and a similar argument for the fourth term yields:
\bean
 \tilde j -k_n + 1< j \leq j' \quad \mbox{gives} \quad 	c^n_{ij} = \sum_{l_1 = 1}^{k_n} \sum_{l_2 = i+1 - [n_k f_k (t^{l}_{j+k_n-l_1})]}^{k_n} g(l_1/k_n) g(l_2/k_n), &&\\
j' < j \leq \bar j +k_n - 1 \quad \mbox{gives} \quad 	c^n_{ij} = \sum_{l_1 = 1}^{k_n} \sum_{l_2 = 1}^{k_n + i -1- [n_k f_k (t^{l}_{j-l_1})]} g(l_1/k_n) g(l_2/k_n).
\eean
Also,
\bean
\tilde j -k_n + 1< j \leq j' \Rightarrow k_n +i -1 - [n_k f_k (t^{l}_{j-l_1})] \geq k_n +i -1- [n_k f_k (t^{l}_{j'-1})] \geq k_n,
\eean
and with the same reasoning for the second case we obtain the global formula
\bean
c^n_{ij} = \sum_{l_1 = 1}^{k_n} \sum_{l_2 = i +1- [n_k f_k (t^{l}_{j+k_n-l_1})]}^{k_n + i -1- [n_k f_k (t^{l}_{j-l_1})]} g(l_1/k_n) g(l_2/k_n).
\eean
In order to simplify this expression further, we use the uniform approximation
\bea \label{app}
n_k f_k (t^{l}_{j+k_n-l_1}) &=& n_k f_k (t^{l}_{j'}) + n_k f'_k (t^{l}_{j'}) (t^{l}_{j+k_n-l_1} - t^{l}_{j'}) + o(k_n) \nonumber \\ &=& n_k f_k (t^{k}_{i}) + n_k f'_k (t^{k}_{i}) (f_l^{-1}((j+k_n-l_1)/n_l) - f_l^{-1}(j'/n_l)) + o(k_n) \nonumber \\ &=& i + h_{kl} (t^{k}_{i}) (j+k_n-l_1-j') + o(k_n).
\eea
From Lemma \ref{prooflem1}, $|j+k_n-l_1 - j'| \leq C k_n$ holds, thus continuity of $f_k$ and its first derivative justifies each approximation. In the same way,
$
n_k f_k (t^{l}_{j-l_1}) = i + h_{kl} (t^{k}_{i}) (j-l_1-j') + o(k_n),
$
and we get
\bean
c^n_{ij} = \sum_{l_1 = 1}^{k_n} \sum_{l_2 = h_{kl} (t^{k}_{i}) (j'-j-k_n+l_1)}^{k_n + h_{kl} (t^{k}_{i}) (j'-j+l_1)} g(l_1/k_n) g(l_2/k_n) +o(k_n^2) = k_n^2 \int_0^1 \int_{h_{kl} (t^{k}_{i}) (\frac{j'-j}{k_n}-1+u)}^{1 + h_{kl} (t^{k}_{i}) (\frac{j'-j}{k_n}+u)} g(u) g(v) dv du +o(k_n^2).
\eean
The claim can now be concluded easily.\qed \\ \\
With the aid of the preceding lemma it is easy to compute the main part of the variance due to Brownian motion. Recall (\ref{gammafun}) and the definition of $\tilde F_{z,p}(k,l)$. Set also
\bean
k(z,p) = [n_k f_k(\frac{z(p+3b)k_n}{n})]+1, \quad \tilde k(z,p) = [n_k f_k(\frac{z(p+b)k_n+(p-2b)k_n}{n})]
\eean
for any $k$, so $t_{k(z,p)}^k$ (or $t_{\tilde k(z,p)}^k$) is usually the smallest (or the largest) point in the $k$th grid which lies within $\tilde B_z(p)$. Then we obtain the following result.

\begin{lem} \label{form}
For any fixed $p$ we have
\bea \nonumber
&& \sum_{(i,j) \in \tilde F_{z,p}(k,l), (q,r) \in F_{z,p}(k',l')} c^n_{ij}(k,l) c^n_{qr}(k',l') \E[\De_i^{n_k} W^{v_1} \De_q^{n_{k'}} W^{v_1}] \E[\De_j^{n_{l}} W^{v_2} \De_r^{n_{l'}} W^{v_2}]
\\ &=& (p-4b) \frac{k_n^6}{n^2} \gamma_{k,l,k',l'}(t_{k(z,p)}^k)  + o(k_n^2),
\eea
uniformly in $z$.
\end{lem}
\textit{Proof.}~ The reason for restricting $(i,j)$ to the set $\tilde F_{z,p}(k,l)$ is that it allows us to use Lemma \ref{cij} to obtain approximate representations for all $c^n_{ij}(k,l)$ and $c^n_{qr}(k',l')$ that correspond to non-zero terms in the left hand side of the statement. In fact, since $t_i^k$ is within $\tilde B_z(p)$, we know from Lemma \ref{prooflem1} that (essentially) any $t_q^{k'}$ with a non-vanishing $\E[\De_i^{n_k} W^{v_1} \De_q^{n_{k'}} W^{v_1}]$ lies within $\tilde B_z(p)$ as well, and thus the conditions for an application of Lemma \ref{cij} are satisfied. We obtain
\bean
&&\sum_{(i,j) \in \tilde F_{z,p}(k,l), (q,r) \in F_{z,p}(k',l')} c^n_{ij}(k,l) c^n_{qr}(k',l') \E[\De_i^{n_k} W^{v_1} \De_q^{n_{k'}} W^{v_1}] \E[\De_j^{n_{l}} W^{v_2} \De_r^{n_{l'}} W^{v_2}] \\ &=& \sum_{i=k(z,p)}^{\tilde k(z,p)} \sum_{j= [n_l f_l(t_{i-k_n}^{k})] - k_n + 1}^{[n_l f_l(t_{i+k_n}^{k})]+k_n - 1} c^n_{ij}(k,l) \sum_{q= [n_{k'} f_{k'}(t_{i-1}^{k})]+ 1}^{[n_{k'} f_{k'}(t_{i}^{k})]+ 1}  (t_i^{k} \wedge t_q^{k'} - t_{i-1}^{k} \vee t_{q-1}^{k'} ) \\ &&\hspace{1cm}  \sum_{r= [n_{l'} f_{l'}(t_{j-1}^{l})]+ 1}^{[n_{l'} f_{l'}(t_{j}^{l})]+ 1} c^n_{qr}(k',l')(t_j^{l} \wedge t_r^{l'} - t_{j-1}^{l} \vee t_{r-1}^{l'} ) + o(k_n^2),
\eean
since both expectations vanish for other choices of $q$ and $r$. Using (\ref{c}) plus continuity of $\psi$ and
$
n_{l'} f_{l'} (t_{j}^{l}) = n_{l'} f_{l'} (t_{i}^{k}) + h_{l'l}(t_{i}^{k}) (j - n_l f_l(t_i^{k})) + o(k_n),
$
which can be shown in the same way as (\ref{app}), we get
\bean
c^n_{qr}(k',l') = c^n_{[n_{k'} f_{k'}(t_{i}^{k})][n_{l'} f_{l'}(t_{j}^{l})]}(k',l') + o(k_n^2) = k_n^2 \psi\Big( h_{l'l}(t_i^{k}) \frac{n_l f_l(t_i^{k})-j}{k_n}, h_{k'l'}(t_i^{k}) \Big) + o(k_n^2).
\eean
Using this approximation, we lose dependence of $c^n_{qr}(k',l')$ on $q$ and $r$. We conclude
\bean
&&\sum_{(i,j) \in \tilde F_{z,p}(k,l), (q,r) \in F_{z,p}(k',l')} c^n_{ij}(k,l) c^n_{rs}(k',l') \E[\De_i^{n_k} W^{v_1} \De_q^{n_{k'}} W^{v_1}] \E[\De_j^{n_{l}} W^{v_2} \De_r^{n_{l'}} W^{v_2}] \\ &=& k_n^4 \sum_{i=k(z,p)}^{\tilde k(z,p)} (t_i^{k}- t_{i-1}^{k}) \sum_{j= [n_l f_l(t_{i-k_n}^{k})] - k_n + 1}^{[n_l f_l(t_{i+k_n}^{k})]+k_n - 1} (t_j^{l} - t_{j-1}^{l} ) \\ &&\hspace{1cm} \psi\Big( \frac{n_l f_l(t_i^{k})-j}{k_n}, h_{k,l}(t_i^{k}) \Big) \psi\Big( h_{l'l}(t_i^{k}) \frac{n_l f_l(t_i^{k})-j}{k_n}, h_{k'l'}(t_i^{k}) \Big) + o(k_n^2).
\eean
Again a Taylor expansion gives
\bee \label{tayl}
t_j^{l} - t_{j-1}^{l} = \frac{1}{n_l f'_l(t_i^{k})} + o(n^{-1})
\eee
and similarly for $t_i^{k} - t_{i-1}^{k}$, and using (\ref{app}) once more we obtain
$
n_l f_l(t_{i+k_n}^{k}) = n_l f_l(t_{i}^{k}) + h_{lk}(t_{i}^{k}) k_n + o(k_n)
$
plus a similar result for $t_{i-k_n}^{k}$. Thus a Riemann approximation and continuity of all functions involved give
\bean
&&\sum_{(i,j) \in \tilde F_{z,p}(k,l), (q,r) \in F_{z,p}(k',l')} c^n_{ij}(k,l) c^n_{rs}(k',l') \E[\De_i^{n_k} W^{v_1} \De_q^{n_{k'}} W^{v_1}] \E[\De_j^{n_{l}} W^{v_2} \De_r^{n_{l'}} W^{v_2}] \\ &=& \frac{k_n^5}{n^2} \sum_{i=k(z,p)}^{\tilde k(z,p)} \frac{1}{m_k f'_k(t_i^k)} \gamma_{k,l,k',l'}(t_{i}^k)  + o(k_n^2) = \frac{k_n^5}{n^2} \frac{\tilde k(z,p) - k(z,p)}{m_k f'_k(t_{k(z,p)}^k)} \gamma_{k,l,k',l'}(t_{k(z,p)}^k)  + o(k_n^2).
\eean
The claim follows now from yet another Taylor expansion. \qed \\ \\
Lemma \ref{form} only gives information about the variance part coming from those $t_i^k$ which belong to $\tilde B_z(p)$. For a fixed $p$ the other terms are not negligible, and in order to prove Theorem \ref{proofth1} it is necessary to show convergence of their contribution to $\E[\beta_{ijqr}^{klk'l'} (1,p) | \mathcal F_{\min B_z(p)}]$ as well. This is why we need two additional results on their asymptotic behavior, which of course are similar in spirit to the preceding ones. Set
\bean
\bar k(z,p) = [n_k f_k(\frac{z(p+b)k_n}{n})]+1, \quad \hat k(z,p) = [n_k f_k(\frac{z(p+b)k_n+pk_n}{n})]
\eean
and let $\tilde F^c_{z,p}(k,l)$ be the complement of $\tilde F_{z,p}(k,l)$ in $F_{z,p}(k,l)$. As an analogue of the function $\psi$ we define
\bean
\vartheta (s,x,y_1,y_2,y_3,y_4)=  \int_{y_1}^{y_2} \int_{\max\{(u-1+s)x,y_3\}}^{\min\{1+x(s+u),y_4\}} g(u) g(v) dv du
\eean
also.
\begin{lem} \label{cij2} Assume $(i,j) \in \tilde F^c_{z,p}(k,l)$. Then for any non-zero $c^n_{ij}(k,l)$ we have the uniform approximation
\bea \label{c2}
c^n_{ij}(k,l) = k_n^2 \vartheta \Big( \frac{n_l f_l(t_i^{k})-j}{k_n}, h_{kl}(t_i^{k}),\frac{j-\hat l(z,p)}{k_n},\frac{j-\bar l(z,p)}{k_n},\frac{i-\hat k(z,p)}{k_n},\frac{i-\bar k(z,p)}{k_n} \Big) + o(k_n^2).
\eea
\end{lem}
Lemma \ref{cij2} can obviously be proven in the same way as Lemma \ref{cij} (but with some more cases to distinguish between), and the only differences between both representations are the extra conditions on the bounds of the integrals, which arise naturally since $c^n_{ij}(k,l)$ is computed at the boundary of $B_z(p)$.

Finally, we need some additional notation. We set
\bean
&&\rho_{k,l,k',l'}(w,x) = \frac{1}{n_l f'_l(w)} \int_{-(1+h_{lk}(w))}^{h_{lk}(w) x} \vartheta\Big(s,h_{kl}(w),0,h_{lk}(w)x-s,0,x\Big) \\ &&\hspace{3cm} \vartheta(h_{l'l}(w) \Big(s,\frac{h_{k'l'}(w)}{h_{l'l}(w)},0,h_{lk}(w)x-s,0,x\Big)) ds
\eean
and
\bean
&&\lambda_{k,l,k',l'}(w,x) = \frac{1}{n_l f'_l(w)} \int_{h_{lk}(w) x-1}^{(1+h_{lk}(w))} \vartheta\Big(s,h_{kl}(w),h_{lk}(w)x-s,1,x,1\Big) \\ &&\hspace{3cm} \vartheta(h_{l'l}(w) \Big(s,\frac{h_{k'l'}(w)}{h_{l'l}(w)},h_{lk}(w)x-s,1,x,1\Big)) ds.
\eean
\begin{lem} \label{form2}
We have
\bea \nonumber
&& \sum_{(i,j) \in \tilde F^c_{z,p}(k,l), (q,r) \in F_{z,p}(k',l')} c^n_{ij}(k,l) c^n_{qr}(k',l') \E[\De_i^{n_k} W^{v_1} \De_q^{n_{k'}} W^{v_1}] \E[\De_j^{n_{l}} W^{v_2} \De_r^{n_{l'}} W^{v_2}] \\ &=& \label{form3} \frac{k_n^6}{n^2} \Big( \frac{1}{m_k f'_k(t^k_{\bar k(z,p)})} \int_0^{2b m_k f'_k(t^k_{\bar k(z,p)})} \rho_{k,l,k',l'}(t^k_{\bar k(z,p)},x) dx \\ &&\hspace{1cm}+ \frac{1}{m_k f'_k(t^k_{\tilde k(z,p)})} \int_{-2b m_k f'_k(t^k_{\tilde k(z,p)})}^1 \lambda_{k,l,k',l'}(t^k_{\tilde k(z,p)},x) dx \Big) + o(k_n^2), \label{form4}
\eea
uniformly in $z$.
\end{lem}
\textit{Proof.}~ Without loss of generality we prove the result for $\bar k(z,p) \leq i < k(z,p)$ only. Note by assumption on $b$ and $g$ that (\ref{c2}) reduces to
\bean
c^n_{ij}(k,l) = k_n^2 \vartheta \Big( \frac{n_l f_l(t_i^{k})-j}{k_n}, h_{kl}(t_i^{k}),0,\frac{j-\bar l(z,p)}{k_n},0,\frac{i-\bar k(z,p)}{k_n} \Big) + o(k_n^2)
\eean
in this case. Mimicking the proof of Lemma \ref{form} the variance part due to these terms becomes
\bean
\quad U^{k,l,k',l'}_{z,p} = \sum_{i=\bar k(z,p)}^{k(z,p)} (t_i^{k}- t_{i-1}^{k}) \sum_{j= \bar l(z,p)}^{[n_l f_l(t_{i+k_n}^{k})]+k_n - 1} (t_j^{l} - t_{j-1}^{l} ) c^n_{ij}(k,l) c^n_{[n_{k'} f_{k'}(t_{i}^{k})][n_{l'} f_{l'}(t_{j}^{l})]}(k',l'),
\eean
up to an error of order $o(k_n^2)$. A similar Taylor expansion as (\ref{app}) gives
\bean
&& c^n_{[n_{k'} f_{k'}(t_{i}^{k})][n_{l'} f_{l'}(t_{j}^{l})]}(k',l') \\ &=& k_n^2 \vartheta(h_{l'l}(t_i^{k}) \Big( \frac{n_l f_l(t_i^{k})-j}{k_n}, \frac{h_{k'l'}(t_i^{k})}{h_{l'l}(t_i^{k})}, 0,\frac{j-\bar l(z,p)}{k_n},0,\frac{i-\bar k(z,p)}{k_n} \Big) ) + o(k_n^2).
\eean
Using (\ref{tayl}) and a Riemann sum argument we obtain
\bean
&&U^{k,l,k',l'}_{z,p} = k_n^5 \sum_{i=\bar k(z,p)}^{k(z,p)} (t_i^{k}- t_{i-1}^{k}) \frac{1}{n_l f'_l(t_i^k)} \int_{-(1+h_{lk}(t_i^k))}^{\frac{n_l f_l (t_i^k)-\bar l(z,p)}{k_n}} \\ && \hspace{2cm} \vartheta \Big( s, h_{kl}(t_i^{k}),0,\frac{n_l f_l (t_i^k)-\bar l(z,p)}{k_n}-s,0,\frac{i-\bar k(z,p)}{k_n} \Big) \\ && \hspace{2cm} \vartheta ( h_{l'l}(t_i^k)\Big( s, \frac{h_{k'l'}(t_i^{k})}{h_{l'l}(t_i^k)},0,\frac{n_l f_l (t_i^k)-\bar l(z,p)}{k_n}-s,0,\frac{i-\bar k(z,p)}{k_n} \Big) ) ds + o(k_n^2).
\eean
The final step differs from the previous proof, as the dependence on $i$ is more involved now. We use continuity to obtain
\bean
\frac{n_l f_l (t_i^k)-\bar l(z,p)}{k_n} = \frac{n_l f_l (t_i^k)- n_l f_l (t^k_{\bar k(z,p)})}{k_n} + o(1) = h_{lk}(t^k_{\bar k(z,p)}) \frac{i- \bar k(z,p)}{k_n} + o(1),
\eean
and applying (\ref{tayl}) on $(t_i^{k}- t_{i-1}^{k})$ plus replacing each $t_i^k$ by $t^k_{\bar k(z,p)}$ due to continuity again, we derive
\bean
U^{k,l,k',l'}_{z,p} &=& \frac{k_n^6}{n^2} \frac{1}{m_k f'_k(t^k_{\bar k(z,p)})} \int_0^{\frac{k(z,p) - \bar k(z,p)}{k_n}} \rho_{k,l,k',l'}(t^k_{\bar k(z,p)},x) dx + o(k_n^2).
\eean
The claim can now be obtained easily.
\qed \\ \\
It is obviously possible to replace $\tilde k(z,p)$ and $\bar k(z,p)$ in (\ref{form3}) and (\ref{form4}) by $k(z,p)$ without affecting the approximation error. We set
\bean
\vp_{k,l,k',l'}(p,w) &=& (p-4b) \gamma_{k,l,k',l'}(w) + \frac{1}{m_k f'_k(w)} \int_0^{2b m_k f'_k(w)} \rho_{k,l,k',l'}(w,x) dx \\ &&\hspace{3cm}+ \frac{1}{m_k f'_k(w)} \int_{-2b m_k f'_k(w)}^1 \lambda_{k,l,k',l'}(w,x) dx,
\eean
and it is simple now to derive the following theorem which concludes this section.
\begin{theo}
We have
\bean
V_{p}^{kl,k'l'}(1) &=& \sum_{z} \E[\zeta_{zn}^{kl} (1,p) \zeta_{zn}^{k'l'} (1,p)| \mathcal F_{\min B_z(p)}] \\ &=& \frac{\theta}{p \psi^4} \int_0^1 \Big( \vp_{k,l,k',l'}(p,w) \Sigma_w^{kk'} \Sigma_w^{ll'} + \vp_{k,l,l',k'}(p,w) \Sigma_w^{kl'} \Sigma_w^{lk'} \Big) dw + o_{\mathbb P} (1).
\eean
For $p \to \infty$, we conclude
\bean
V_{p}^{kl,k'l'}(1) \toop \frac{\theta}{\psi^4} \int_0^1 \Big( \gamma_{k,l,k',l'}(w) \Sigma_w^{kk'} \Sigma_w^{ll'} + \gamma_{k,l,l',k'}(w) \Sigma_w^{kl'} \Sigma_w^{lk'} \Big) dw,
\eean
which equals the pure diffusion part of (\ref{variance}).
\end{theo}

\subsection{The contribution of the remaining parts to the variance} \label{proofth1part2}
In this final subsection we give some ideas on how to obtain formulas for $V_{p}^{kl,k'l'}(2)$ and $V_{p}^{kl,k'l'}(3)$, from which Theorem \ref{proofth1} (and thus in turn Theorem \ref{th1}) can be concluded.

The main intuition in both cases it that one obtains representations for $\alpha_{ij}^{kl} (2,p)$ and $\alpha_{ij}^{kl} (3,p)$ which are closely related to (\ref{cdef}) in the sense that those constants $\overline{c}^n_{ij}(k,l)$ and $\widetilde{c}^n_{ij}(k,l)$, say, can be treated in the same way as in Lemma \ref{cij} and Lemma \ref{cij2}. In fact, the only difference is that $g(l_1/k_n)$ sometimes has to be replaced by $(-1/k_n) g'(l_1/k_n)$, since $(g(l_1/k_n) - g((l_1+1)/k_n)) \varepsilon_{t_j^l}^l$ plays the role of $g(l_1/k_n) \De_j^{n_l} W$ now, and so the approximating functions in a version of Lemma \ref{cij} naturally become $\overline \psi$ and $\widetilde \psi$ from (\ref{psifun}).

Also, Lemma \ref{form} and Lemma \ref{form2} have expressions in this context, but the first difference is that one does not sum over all $(i,j)$ and $(p,q)$ now, but only over those for which ${t_j^l}$ and $t_r^{l'}$, say, coincide, as otherwise $\E [\varepsilon_{t_j^l}^l \varepsilon_{t_r^{l'}}^{l'}] \neq 0$ is not satisfied. Second,
\bean
t_j^{l} - t_{j-1}^{l} = \frac{1}{n_l f'_l(t_i^{k})} + o(n^{-1})
\eean
is not included in the sum anymore, as this term came from an increment of Brownian motion. This explains the need for the additional terms $m_{ll'} f'_{ll'}$ in $\overline \gamma$ and $\widetilde \gamma$, as the Riemann approximation otherwise does not hold.

\pagebreak

\end{document}